\documentclass[aps,twocolumn, superscriptaddress,amsmath,amssymb]{revtex4-1}

\usepackage[utf8]{inputenc}
\usepackage[T1]{fontenc}
\usepackage{times}

\usepackage{enumitem}
\usepackage{graphicx}
\usepackage{bm}
\usepackage{layout}
\usepackage{float}
\usepackage{physics,siunitx}
\usepackage{bbm,mathtools}
\usepackage{amsthm,amssymb,amsmath,amsfonts,bbm}
\usepackage{MnSymbol} 
\usepackage{xcolor}
\usepackage{soul}

\usepackage{hyperref}
\hypersetup{
  colorlinks=true,
  hypertexnames=false,
  linktocpage,
  colorlinks=true, 
  urlcolor=magenta!90!black,    
  linkcolor=blue!60!black, 
  citecolor=black!60 
}

\usepackage{csquotes}

\newtheorem{theorem}{Theorem}

\newtheorem{lemma}{Lemma}

\newtheorem{conjecture}{Conjecture}

\usepackage[capitalise,compress]{cleveref}

\usepackage{braket}
\DeclareMathOperator*{\E}{\mathbb{E}}
\newcommand{\per}{\mathrm{Per}}
\newcommand{\haf}{\mathrm{Haf}}

\usepackage{array}
\let\vec\mathbf

\usepackage{mathtools}
\DeclarePairedDelimiter\ceil{\lceil}{\rceil}

\DeclarePairedDelimiter\bkt{[}{]}		
\DeclarePairedDelimiter\paren{(}{)}

\usepackage{xprintlen}
\begin{document}

\title{Transition of Anticoncentration in Gaussian Boson Sampling}

\author{Adam Ehrenberg}
\affiliation{Joint Center for Quantum Information and Computer Science, NIST/University of Maryland College Park, Maryland 20742, USA}
\affiliation{Joint Quantum Institute, NIST/University of Maryland College Park, Maryland 20742, USA}

\author{Joseph T. Iosue}
\affiliation{Joint Center for Quantum Information and Computer Science, NIST/University of Maryland College Park, Maryland 20742, USA}
\affiliation{Joint Quantum Institute, NIST/University of Maryland College Park, Maryland 20742, USA}

\author{Abhinav Deshpande}
\affiliation{IBM Quantum, Almaden Research Center, San Jose, California 95120, USA}

\author{Dominik Hangleiter}
\affiliation{Joint Center for Quantum Information and Computer Science, NIST/University of Maryland College Park, Maryland 20742, USA}

\author{Alexey V. Gorshkov}
\affiliation{Joint Center for Quantum Information and Computer Science, NIST/University of Maryland College Park, Maryland 20742, USA}
\affiliation{Joint Quantum Institute, NIST/University of Maryland College Park, Maryland 20742, USA}
\date{\today}

\begin{abstract}
Gaussian Boson Sampling is a promising method for experimental demonstrations of quantum advantage because it is easier to implement than other comparable schemes. 
While most of the properties of Gaussian Boson Sampling are understood to the same degree as for these other schemes, we understand relatively little about the statistical properties of its output distribution.
The most relevant statistical property, from the perspective of demonstrating quantum advantage, is the \emph{anticoncentration} of the output distribution as measured by its second moment. 
The degree of anticoncentration features in arguments for the complexity-theoretic hardness of Gaussian Boson Sampling, and it is also important to know when using cross-entropy benchmarking to verify experimental performance. 
In this work, we develop a graph-theoretic framework for analyzing the moments of the Gaussian Boson Sampling distribution. 
Using this framework, we show that Gaussian Boson Sampling undergoes a transition in anticoncentration as a function of the number of modes that are initially squeezed compared to the number of photons measured at the end of the circuit. When the number of initially squeezed modes scales sufficiently slowly with the number of photons, there is a lack of anticoncentration. 
However, if the number of initially squeezed modes scales quickly enough, the output probabilities anticoncentrate weakly.
\end{abstract}

\maketitle
There is a hope that quantum computation will be able to outperform classical computation on certain tasks. 
In particular, there has been a recent explosion of interest in so-called sampling problems given the strong theoretical evidence for an exponential speedup of quantum algorithms over the best possible classical algorithms; see Ref.~\cite{hangleiterComputationalAdvantageQuantum2023} for an overview. 
Aaronson and Arkhipov introduced one of the most deeply studied sampling frameworks in their seminal work on Boson Sampling \cite{aaronsonComputationalComplexityLinear2013}. 
The Boson Sampling task is to approximately sample from the outcome distribution of measuring $n$ single photons in $m$ optical modes transformed by a Haar-random linear-optical unitary, which can be implemented as a random network of beamsplitters and phase shifters \cite{reck_experimental_1994}. 
\textcite{aaronsonComputationalComplexityLinear2013} focused on single-photon input states, but these can be challenging to produce experimentally \cite{wang_boson_2019} because existing single-photon sources are not sufficiently reliable to avoid an exponential amount of post-selection \cite{brod_photonic_2019}. Therefore, there has been an interest in generalizing the original Boson Sampling setup to other input states.

The currently most feasible generalization is Gaussian Boson Sampling (GBS) \cite{lundBosonSamplingGaussian2014,rahimi-keshariWhatCanQuantum2015,hamiltonGaussianBosonSampling2017,kruseDetailedStudyGaussian2019a}, which uses Gaussian input states. 
These states are significantly easier to prepare reliably than single-photon states.
At the same time, similar statements can be made about the hardness of sampling from the corresponding output distribution \cite{kruseDetailedStudyGaussian2019a,
deshpandeQuantumComputationalAdvantage2022a,grierComplexityBipartiteGaussian2022,chabaudResourcesBosonicQuantum2023b}, 
and several large-scale GBS experiments have been performed recently \cite{zhongQuantumComputationalAdvantage2020a, zhongPhaseProgrammableGaussianBoson2021a,madsen_quantum_2022,deng_gaussian_2023}. 

Broadly speaking, the hardness of Boson Sampling is based on the connection between output probabilities and the permanent, which is, classically, $\#\mathsf{P}$-hard to compute exactly~\cite{valiant_complexity_1979}. 
Similarly, the hardness of GBS arises from the fact that output probabilities are controlled by a generalization of the permanent called the hafnian; while the permanent counts the number of perfect matchings in a weighted bipartite graph, the hafnian counts the number of perfect matchings in an arbitrary weighted graph~\cite{barvinokCombinatoricsComplexityPartition2016}. Because the hafnian generalizes the permanent, it is also difficult to compute classically. 

However, the complexity of classically computing an individual output probability defined in terms of the permanent or the hafnian is not itself sufficient to prove hardness of sampling from the overall probability distributions. 
The standard hardness argument based on Stockmeyer's algorithm \cite{stockmeyer_complexity_1983,hangleiterComputationalAdvantageQuantum2023} requires that outcome probabilities of \emph{random} Boson Sampling instances be hard to \emph{approximate}. 
Jointly with provable hardness of nearly exactly computing output probabilities~\cite{deshpandeQuantumComputationalAdvantage2022a}, anticoncentration of the outcome probabilities serves as evidence for this. 
Intuitively, if most outcome probabilities are comparable to the uniform probability, then a good classical sampling algorithm needs a very precise idea of each probability's relative magnitude because all of them are important. 
Anticoncentration quantifies this idea as, most concisely, the outcome-collision probability (i.e.\ the probability of getting the same outcome from two independent samples) of the GBS distribution averaged over the choice of linear-optical unitary and normalized by the size of the sample space \cite[][Sec. D]{hangleiterComputationalAdvantageQuantum2023}. 
While a weak form of anticoncentration holds in Boson Sampling \cite{aaronsonComputationalComplexityLinear2013}, under what conditions anticoncentration holds in GBS is an open question. 

In this work, we analyze the moments of GBS in the photon-collision-free limit. 
In this limit, the output distribution is dominated by outcomes with at most a single photon in each mode, and the moments of GBS approximately reduce to moments of squared hafnians of Gaussian random matrices. 
We show that evaluating those moments reduces to counting the connected components of certain graphs.
Using this perspective, we find a closed-form expression for the first moment and derive analytical properties of the second moment. 
We then identify a transition in anticoncentration in GBS:
when the number of initially squeezed modes is large enough compared to the measured number of photons~$n$, a weak version of anticoncentration holds where the normalized average outcome-collision probability scales as~$\sqrt n$. 
However, when sufficiently few modes are initially squeezed, there is a lack of anticoncentration, as the normalized second moment scales exponentially in~$n$. 

The rest of the paper proceeds as follows. We first provide background information and set up the system and problem of interest. We then derive the graph-theoretic formalism for computing the first moment of the output probabilities. 
We proceed to discuss how to apply the formalism to calculate certain properties of the second moment. These results let us finally prove the transition in anticoncentration.

\begin{figure}
\includegraphics{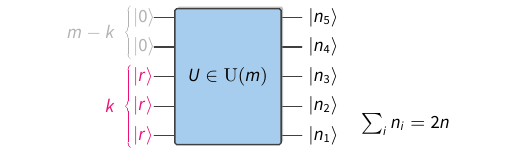}
    \caption{\label{fig:gbs} In Gaussian Boson Sampling (GBS), $k$ out of $m$ modes are prepared in single-mode squeezed states with parameter $r$, while the remaining modes are prepared in the vacuum state $\ket 0$. The modes are then transformed by a Haar random linear-optical unitary $U$ and measured in the Fock basis with outcome counts $n_i$ summing to $2n$. 
    }
\end{figure}

\textit{Setup.---}We consider a photonic system with $m$ modes that is transformed by a Haar-random linear optical unitary $U \in \mathrm{U}(m)$ acting on the modes of the system, see \cref{fig:gbs}. 
In the paradigmatic version of GBS \cite{hamiltonGaussianBosonSampling2017,kruseDetailedStudyGaussian2019a}, the first $k$ of the $m$ modes are prepared in single-mode squeezed states with equal squeezing parameter $r$, while the remaining $m-k$ modes are prepared in the vacuum state.
After applying $U$, all $m$ modes are measured in the Fock basis.

Reference~\cite{hamiltonGaussianBosonSampling2017} proves that, given a unitary $U$, the probability of obtaining an outcome count vector $\vec{n} = (n_1, n_2, \ldots,  n_m) \in \mathbb N_{0}^m$ with total photon count $2n = \sum_{i=1}^{m} n_i $ is given by
\begin{equation}\label{eqn:gbs_probability}
    P_U(\vec n) = \frac{\tanh^{2n} r}{\cosh^{k} r}\abs{\haf(U^{\top}_{1_{k},\vec{n}}U_{1_{k},\vec{n}})}^{2}, 
\end{equation}
where $U_{1_k,\vec n}$ denotes the $k \times 2n$ submatrix of $U$ given by its first $k $ rows and the columns selected according to the nonzero entries of $\vec n$ each copied $n_i$ times \footnote{
Note that squeezed states are supported only on even Fock states, so the total photon count $2n$ must always be even.}.
Moreover,
\begin{equation}\label{eqn:hafnian}
    \haf(A) = \frac{1}{2^{n}n!}\sum_{\sigma \in S_{2n}}\prod_{j = 1}^{n}A_{\sigma(2j-1),\sigma(2j)}
\end{equation}
denotes the hafnian of a $2n\times2n$ symmetric matrix $A$ \footnote{Other equivalent definitions exist, but \cref{eqn:hafnian} is the most convenient one for our purposes.}.

We will work in the regime in which the output states are, with high probability, (photon-)collision-free, meaning that $n_{i} \in \{0,1\}$ for all $i$. A sufficient condition for this to be the case is that the expected number of photons  $\E[2n] = k \sinh^{2}(r) = o(\sqrt{m})$.
In this regime, Ref.~\cite{deshpandeQuantumComputationalAdvantage2022a} provides evidence that, for any fixed $n = o(\sqrt m)$, the distribution over submatrices is well-captured by a generalization of the circular orthogonal ensemble (COE) \footnote{Note that, strictly speaking, the conjecture is only formulated for the regime $n \leq k$ in Ref.~\cite{deshpandeQuantumComputationalAdvantage2022a}. 
However, the evidence provided there for the case $k=n$---Ref.~\cite{aaronsonComputationalComplexityLinear2013} proves that $n \times n$ submatrices of Haar-random unitaries are approximately Gaussian---clearly also holds for the case $k \leq n$.}. 
\begin{conjecture}[Hiding \cite{deshpandeQuantumComputationalAdvantage2022a}] \label{con:gaussian}
For any $k$ such that $1 \leq k \leq m$ and $2n=o(\sqrt{m})$, the distribution of the symmetric product $U^{\top}_{1_{k},\vec{n}}U_{1_{k},\vec{n}}$ of submatrices of a Haar-random $U \in \mathrm{U}(m)$ closely approximates in total variation distance the distribution of the symmetric product $X^{\top}X$ of a complex Gaussian matrix $X~\sim ~\mathcal{N}(0,1/m)_{c}^{k \times 2n}$ with mean $0$ and variance $1/m$.
\end{conjecture}

We work under the assumption that Conjecture~\ref{con:gaussian} is true. 
Fixing the measured number of photons $2n$, the normalized average (outcome-)collision probability, which quantifies anticoncentration, can be written as $|\Omega_{2n}| \E_{U \in U(m)} [\sum_{\vec n \in \Omega_{2n}} P_U(\vec n)^2] $, where $|\Omega_{2n}|$ is the size of the photon-non-collisional sample space of $2n$ photons in $m$ modes, which is the dominant space of outputs when $n = o(\sqrt{m})$. 
Conjecture~\ref{con:gaussian} implies that, with respect to random choices of $U$, all outcomes are equally distributed over the unitaries, the so-called hiding property. 
This implies that the inverse size of the  sample space is given by the first moment $\mathbb E_U [P_U(\vec n)]$. See \cref{sec:asymptotics} of the Supplemental Material (SM) for more details.  
Under Conjecture~\ref{con:gaussian}, the anticoncentration property thus reduces to computing the moments
\begin{equation}\label{eqn:moment_def}
    M_t(k,n) \coloneqq \E_{X \sim \mathcal G^{k \times 2n}}[|\haf(X^\top X)|^{2t}] 
\end{equation}
of the squared hafnian as a function of the parameters $k$ and $n$ for $t = 1,2$,
where we have abbreviated $\mathcal{N}(0,1)_{c}^{k\times 2n}$ as  $\mathcal{G}^{k\times 2n}$. 
We consider unit variance because rescaling $X$ by $1/\sqrt m$ just leads to an overall prefactor that, like the prefactor in \cref{eqn:gbs_probability}, is irrelevant to the normalized average outcome-collision probability. 
We will phrase our discussion in terms of the inverse of the average collision probability 
\begin{equation} \label{eqn:m2}
m_2(k,n) \coloneqq M_1(k,n)^2/M_2(k,n). 
\end{equation}

\textit{First moment and graph-theoretic formalism.---}We begin by analyzing the (rescaled) first moment $M_1$ of the output probabilities. 
In order to derive our graph-theoretic formalism, we use \cref{eqn:hafnian} to expand the hafnian in \cref{eqn:moment_def} as a sum over permutations of a product of matrix elements. 
From there, the key is to use that the matrix elements are independent complex Gaussian, meaning that $\mathbb E_{X \sim \mathcal G^{k}}[X_i X^{*}_j] = \delta_{ij}$ and $\mathbb E_{X \sim \mathcal G^{k}}[X_i X_j]=\mathbb E_{X \sim \mathcal G^{k}}[X^{*}_i X^{*}_j]=0$. 
This yields
\begin{equation}\label{eqn:first_moment_delta}
 M_1(k,n) =  \frac{(2n)!}{(2^{n}n!)^{2}}\sum_{\tau \in S_{2n}}\sum_{\{o_{i}\}_{i=1}^{n}}^{k}\prod_{j=1}^{n}\delta_{o_{\ceil*{\frac{\tau(2j-1)}{2}}}o_{\ceil*{\frac{\tau(2j)}{2}}}}. 
\end{equation}
Let us briefly discuss \cref{eqn:first_moment_delta}, see the SM for details. The sum over $\tau \in S_{2n}$ and the product over index $j$ come from \cref{eqn:hafnian}; the sum over the indices $o_i \in [k] \coloneqq \{1, 2, \ldots, k\}$ is due to an expansion of $X^\top X$ as a matrix product. 
Note that, when $\tau(2j-1)$ and $\tau(2j)$ form a tuple $(2\ell-1, 2\ell)$, then the Kronecker $\delta$ always equals $1$ for index $o_{\ell}$, such that summing over $o_{\ell}$ yields a factor of $k$. 
When $\tau(2j-1)$ and $\tau(2j)$ do not form such a tuple, we get a nontrivial relationship between indices that decreases the number of independent degrees of freedom, thus decreasing the number of factors of $k$ in the final answer. Therefore, to evaluate this expression, one must determine the number of ``free indices'' over all the permutations in $S_{2n}$. We accomplish this with our graph-theoretic approach. 

\begin{figure}
    \centering
    \includegraphics{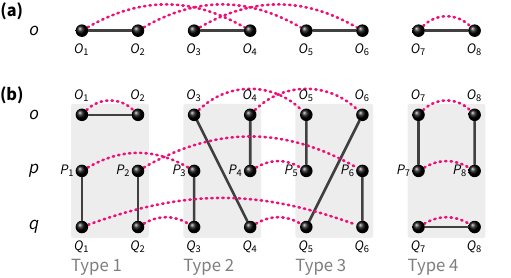}
    \caption{Examples of graphs used to calculate the (a) first and (b) second moments of GBS outcome probabilities. (a) $G \in \mathbb{G}^{1}_{4}$. There are eight vertices $O_{1}$ to $O_{8}$ representing the index $o$ (labeled in the left column). The black (solid) edges connect only adjacent pairs, whereas the red (dashed) edges form an arbitrary perfect matching. This graph has two connected components, meaning it contributes $k^{2}$ to the first moment. (b) $G \in \mathbb{G}_{4}^{2}$. The black (solid) edges are, from left to right, type-1, type-2, type-3, and type-4, as denoted by the gray background.  $z = 1 + 0\times 4^{3} + 1 \times 4^{2} + 2 \times 4^{1} + 3 \times 4^{0} = 28$  (this is calculated by converting 0123 from base-4 into base-10 and then adding 1 such that the final result is in $[4^{4}]$). Note that black (solid) edges stay within two adjacent columns. Red (dashed) edges form an arbitrary perfect matching in each row. This graph contributes $k^{5}$ to the second moment, as there are five connected components.}
    \label{fig:first_moment_example}
    \label{fig:second_moment_red_edges}
\end{figure}

Specifically, define a graph $G_{\tau}$ as follows, see \cref{fig:first_moment_example}(a). 
Let $G_{\tau}$ have $2n$ vertices labeled $O_{1}$ through $O_{2n}$. These upper-case vertices are not directly equivalent to the lower-case indices in \cref{eqn:first_moment_delta}---instead, each index $o_{j}$ splits into two vertices $O_{\ell}$ and $O_{\ell'}$ such that $\ceil{\tau(\ell)/2} = j = \ceil{\tau(\ell')/2}$ (in other words, $o_{\lceil\tau(\ell)/2\rceil}$ maps to a vertex $O_{\ell}$). 
Let $G_{\tau}$ have a black edge between $O_{2j-1}$ and $O_{2j}$ for all $j \in [n]$, and a red edge between $O_{\ell}$ and $O_{\ell'}$ if $\ceil{\tau(\ell)/2}=\ceil{\tau(\ell')/2}$. 
These two kinds of edges mimic two types of ways that dependencies in \cref{eqn:first_moment_delta} can be induced through an index $j$. Red edges identify the $\ell$ and $\ell'$ that map to the same value via $\tau$ and the ceiling function, hence red edges identify which vertices came from the same $o$ index. Black edges identify that \cref{eqn:first_moment_delta} has a Kronecker $\delta$ between $o_{\lceil{\tau(2j-1)}/2\rceil}$ and $o_{\lceil{\tau(2j)}/2\rceil}$.

We see, then, that the number of connected components of $G_{\tau}$, $C(G_{\tau})$, is equivalent to the number of free indices in the sum in \cref{eqn:first_moment_delta}. Therefore, 
\begin{equation}\label{eqn:first_moment_G1}
     M_1(k,n) = \frac{(2n)!}{\paren*{2^{n}n!}^{2}} \sum_{\tau\in S_{2n}}k^{C(G_{\tau})}.
\end{equation}

Now, there is a degeneracy where many permutations induce the same final graph. 
Each graph has the same fixed set of black edges and then one of $(2n-1)!!$ possible sets of red edges (this is the number of ways of pairing $2n$ elements when order does not matter). 
For each graph $G$ corresponding to some assignment of the red edges, there are $2^{n}n!$ permutations $\tau$ such that $G_{\tau} = G$. 
Therefore, instead of studying $G_{\tau}$ as instantiated by permutations $\tau$, we study the underlying graphs $G$. Define $\mathbb{G}^{1}_{n}$ to be the set of graphs on $2n$ vertices with one perfect matching defined by the fixed set of black edges and one perfect matching defined by the arbitrary red edges. 
We can thus rewrite $M_1 = (2n-1)!! \sum_{G\in\mathbb{G}^{1}_{n}}k^{C(G)}$ and state our first result. 
\begin{theorem}\label{thm:first_moment}
The sum over graphs in $\mathbb{G}^{1}_{n}$ satisfies 
\begin{equation}\label{eqn:g_first_moment}
   \sum_{G\in\mathbb{G}^{1}_{n}}k^{C(G)} = k(k+2)\dots (k+2n-2),
\end{equation}
and hence $M_1(k,n) = (2n-1)!!(k+2n-2)!!/(k-2)!!$.
\end{theorem}
The proof proceeds by induction over $n$, where the inductive step reduces a graph in $\mathbb{G}_{n}^{1}$ to one in $\mathbb{G}_{n-1}^{1}$ through an analysis of the red edge connected to $O_{1}$.
There are two options for this red edge: 
either it connects to $O_{2}$, the vertex with which $O_{1}$ shares a black edge, or it attaches to some $O_{x > 2}$. The former creates a connected component of size two; the latter reduces to a graph in $\mathbb{G}_{n-1}^{1}$ by merging vertices $O_{1}$, $O_{2}$, and $O_{x}$ (which does not change the number of connected components). 
Full details can be found in the SM. 

\textit{Second Moment.---}We now sketch the application of our graph-theoretic formalism to the second moment $M_2$, deferring the details to the SM. 
We expand $\abs{\haf (X^{\top}X)}^{4}$ using \cref{eqn:hafnian} which becomes a sum of products of matrix elements that are indexed by four permutations in $S_{2n}$. 
The independence of matrix elements again ensures that we must have an equal number of copies of $X_{ij}$ and $X_{ij}^{*}$ for the expectation value not to vanish on a given product. 
However, because there are more copies of $X$, there are more ways of matching the indices. 
Accounting for these possibilities leads to an expression analogous to \cref{eqn:first_moment_delta}. 
The key differences are the following: 
(1) instead of summing over a single permutation, we now sum over three permutations, labeled $\tau, \alpha, \beta$ (as in the first moment, one of the original four permutations eventually becomes redundant); 
(2) instead of summing over $n$ indices $\{o_{i}\}_{i=1}^{n}$, we now sum over $3n$ indices $\{o_{i}, q_{i}, p_{i}\}_{i=1}^{n}$; 
(3) each factor is a sum of four possible Kronecker $\delta$ terms corresponding to the different types of index matching.

As before, we will define a useful set of graphs; see \cref{fig:second_moment_red_edges}(b) for an example. We expand each index in $\{o_{i}, q_{i}, p_{i}\}_{i=1}^{n}$ to two graph vertices $\{O_{i}, Q_{i}, P_{i}\}_{i=1}^{2n}$, and we organize them into $2n$ columns and three rows assigned to $O$, $P$, and $Q$ vertices, respectively. 
We then use the Kronecker $\delta$s to define black and red edges between these vertices. 
Fixing permutations $\tau, \alpha, \beta$, there is a red edge between  $O_{\ell}$ and $O_{\ell'}$ if $\ceil{\tau(\ell)/2}=\ceil{\tau(\ell')/2}$, and similarly for the $P$ and $Q$ vertices using permutations $\alpha$ and $\beta$, respectively. This means that the red edges are constrained to lie within rows in the graph. Furthermore, these red edges again identify the vertices originating from the same index. 
Because each factor has four Kronecker $\delta$ terms, each factor contributes one of four patterns of black edges, which we refer to as type-1, type-2, type-3, and type-4. Due to the nature of these Kronecker $\delta$ terms, the black edges are constrained to lie within pairs of adjacent columns. Each graph then has one of $4^{n}$ possible sets of black edges indexed by an integer $z$. 
We therefore call these graphs $G_{\tau,\alpha,\beta}(z)$. 

As in the first moment, the number of connected components $C(G_{\tau,\alpha,\beta}(z))$ of the graph $G_{\tau,\alpha,\beta}(z)$ gives the number of free indices of its corresponding term in the expansion of the hafnian, meaning that graph contributes $k^{C(G_{\tau,\alpha,\beta}(z))}$ to the sum. The second moment then becomes
\begin{equation}\label{eqn:second_moment_intermediate}
   M_2(k,n) = \frac{(2n)!}{\paren*{2^{n}n!}^{4}}\sum_{\tau,\alpha,\beta \in S_{2n}} \sum_{z \in [4^{n}]}k^{C(G_{\tau,\alpha,\beta}(z))}.
\end{equation}
We also again use the fact that many permutations induce the same final graph. 
We thus define $\mathbb{G}^{2}_{n}(z)$ to be the set of graphs for the $z$th set of black edges and $\mathbb{G}^{2}_{n} \coloneqq \bigcup_{z\in[4^{n}]}\mathbb{G}^{2}_{n}(z)$. 
Because there are now three permutations associated to each graph, we obtain a degeneracy factor of $(2^{n}n!)^{3}$ and find  
\begin{equation}\label{eqn:second_moment_polynomial}
 M_2(k,n) = (2n-1)!!\sum_{G\in\mathbb{G}^{2}_{n}}k^{C(G)}.
\end{equation}
We can now state our second result.
\begin{theorem}\label{thm:second_moment}
The second moment $M_2(k,n)$ is a degree-$2n$ polynomial in $k$ and can be written as $ M_2(k,n) = \mbox{$(2n-1)!!$} \sum_{i = 1}^{2n} c_{i}k^{i}$, where $c_{i}$ is the number of graphs $G\in \mathbb{G}^{2}_{n}$ that have $i$ connected components. 
\end{theorem}
\noindent\cref{thm:second_moment} follows from \cref{eqn:second_moment_polynomial} and verifying that the limits of summation are correct, which is done in the SM.

\textit{Transition in anticoncentration.---}We now use Theorems~\ref{thm:first_moment} and \ref{thm:second_moment} in order to show that anticoncentration in GBS undergoes a transition as a function of~$k$. 
Roughly speaking, $m_2(k,n)$  upper-bounds the fractional support of the outcome distribution on outcomes with probabilities larger than uniform, i.e.\ those most relevant to the sampling task, as we explain in detail in the SM. 
We speak of strong anticoncentration if $m_2(k,n) \geq \mathrm{const.}$.
We speak of weak anticoncentration if $m_2(k,n) \geq 1/\mathrm{poly}(n)$.
If $m_2(k,n) = O(1/n^{a})$ for any constant $a > 0$, however, we say that there is a lack of anticoncentration; in this regime, only a negligible fraction of the probabilities are nontrivial. 
While our definition of anticoncentration in terms of $m_2$ is stronger than the standard definition, it captures the essence of anticoncentration, see the SM. 
We show a transition between a lack of anticoncentration for $k = 1$ and weak anticoncentration for $k \rightarrow \infty$ (which, of course, requires $m \to \infty$ as well).

In order to do so, we analyze the  polynomial coefficients $c_i$, observing that for $k=1$, $M_2(k,n) = (2n-1)!! \sum_{i=1}^{2n} c_i$, and for $k \rightarrow \infty$ \footnote{For any $n$, there exists some sufficiently large $k$ for which the leading order term dominates. The exact required scaling of $k$ with $n$ is investigated more thoroughly in the companion piece Ref.~\cite{companion}.}, $M_2(k,n) \approx (2n-1)!! \, c_{2n} k^{2n}$. The following lemma states our results for these regimes.
\begin{lemma}\label{lemma:k}
We have that 
\begin{enumerate}
    \item[i.] $M_2(1,n) = ((2n-1)!!)^4 4^n$
    \item[ii.]  $ c_{2n} = (2n)!!$
\end{enumerate}
\end{lemma}
Part (i) of the lemma follows from a simple, direct computation using the expansion of the second moment in terms of Kronecker $\delta$s;
part (ii) follows by reducing the graph counting problem to a special instance of the first moment with $k = 2$. 
This reduction happens because the types of edges that are allowed in order to get $2n$ connected components are quite restrictive, see the SM for details. 

\cref{thm:second_moment} and \cref{lemma:k} imply that, when $k = n^{0} =1$, the inverse normalized second moment is negligible: 
\begin{equation}
    m_2(1,n) = \frac{((2n-1)!!^{2})^{2}}{(2n-1)!!^{4}4^{n}}=4^{-n}. 
\end{equation}
Now take $k = n^{a} \leq m$, for some large $a$. 
In this limit, $M_2(k,n)$ is dominated by the behavior of its leading order in $k$, which is $(2n-1)!! (2n)!! k^{2n}$.
Additionally, $M_1(k,n) = (2n-1)!! (k + 2n-2)!!/(k-2)!! \sim (2n-1)!! k^{n}$ and, hence, the $k$-dependence of $m_2(k,n)$ vanishes. 
Using Stirling's approximation on the remaining $n$-dependence yields
\begin{equation}
   m_2(k,n) \sim \frac{(2n - 1)!!}{(2n)!}  = \frac{(2n)!}{4^{n}(n!)^{2}} \sim 
    \frac{1}{\sqrt{\pi n}}.
\end{equation}
This proves the central claim of our work. In the SM, we also show how anticoncentration of the approximate GBS distribution relates to anticoncentration of the true distribution.

\textit{Discussion and conclusion---}In this Letter, we have shown a transition in anticoncentration in the output probabilities of GBS as a function of the number of initially squeezed modes. 
The presence of anticoncentration is additional evidence for the hardness of GBS, and our results thus yield clear advice for experiments in the collision-free regime: given a desired average photon number, distribute the required squeezing for this number across all modes. 

Our results give rise to an interesting state of affairs when considered in conjunction with the hiding property: 
in both GBS and standard Boson Sampling, the hiding property is known to fail outside of the highly dilute collision-free regime which is characterized by $m = O(n^2)$ \cite{jiang_entries_2009,jiang_how_2006}, while it is conjectured to hold for any $m = \omega(n^2)$ \cite{aaronsonComputationalComplexityLinear2013,deshpandeQuantumComputationalAdvantage2022a}. 
Standard Boson Sampling anticoncentrates weakly with inverse normalized second moment $1/n$ in the same regime \cite[][Lemma 8.8]{aaronsonComputationalComplexityLinear2013}. 
The only relevant scale is thus the relative size of the number of modes to the number of photons. 
In GBS, we now find an additional relevant scale, the number of squeezed modes in the input state. 
This scale does not seem to be relevant to the hiding property in GBS which holds for $m = \omega(n^2) $ and \emph{any} $k$ under Conjecture~\ref{con:gaussian}.
We find, however, that it is very relevant to the anticoncentration property of GBS. 

For a potential explanation of the relevance of this scale, we refer to Scattershot Boson Sampling, which is ``intermediate'' between standard Boson Sampling and GBS.
In Scattershot Boson Sampling, $n$ single photons are distributed randomly across the input modes using postselection on two-mode squeezed states. 
In order to satisfy collision-freeness in the input state with high probability, the total squeezing in the input needs to be distributed across $\omega(n^2)$ initial squeezed states \cite{lundBosonSamplingGaussian2014}, see the SM for details. 
It is not clear to what extent this explanation generalizes to GBS, however, because the distribution of photons in the input state of GBS is only supported on (collision-full) integer multiples of photon pairs in every mode. Therefore, this connection warrants future study.

Our results also connect to the classical simulability of GBS. 
The hafnian of $A$ can be computed in time exponential in the rank of $A$ \cite{bjorklund_faster_2019-1}. 
The absence of anticoncentration for small $k \ll n$ overlaps with this regime of efficient classical simulability, as the rank of $X^\top X$ is upper-bounded by~$k$. 

But does it also extend beyond this regime? 
While we have been able to prove the existence of this transition, our above work is not sufficient to pin down its precise location. 
However, we conjecture that, weak anticoncentration holds for $k = \omega(n^{2})$, but there is a lack of anticoncentration for $k = O(n^{2})$.
In a companion work \cite{companion}, we give evidence for this conjecture by fully analyzing the coefficients $c_i$ of $M_2(k,n)$. 
\vspace{1em}

\let\oldaddcontentsline\addcontentsline
\renewcommand{\addcontentsline}[3]{}
\acknowledgments
\let\addcontentsline\oldaddcontentsline

We are grateful to Changhun Oh for sharing his independent calculation of the first moment.
We also thank Changhun Oh, Bill Fefferman, Marcel Hinsche, Max Alekseyev, and Benjamin Banavige for helpful discussions, and Marcel Hinsche for comments on an earlier version of the manuscript.   
A.E., J.T.I., and A.V.G.~were supported in part by the DoE ASCR Accelerated Research in Quantum Computing program (award No.~DE-SC0020312), DARPA SAVaNT ADVENT, AFOSR MURI, DoE ASCR Quantum Testbed Pathfinder program (awards No.~DE-SC0019040 and No.~DE-SC0024220), NSF QLCI (award No.~OMA-2120757), NSF STAQ program, and AFOSR.  
Support is also acknowledged from the U.S.~Department of Energy, Office of Science, National Quantum Information Science Research Centers, Quantum Systems Accelerator.
J.T.I thanks the Joint Quantum Institute at the University of Maryland for support through a JQI fellowship.
D.H. acknowledges funding from the US Department of Defense through a QuICS Hartree fellowship.

\let\oldaddcontentsline\addcontentsline
\renewcommand{\addcontentsline}[3]{}

\bibliography{doms_bib,LXEB,suppmat}

\let\addcontentsline\oldaddcontentsline

\clearpage

\onecolumngrid
\cleardoublepage
\setcounter{page}{1}

\setcounter{section}{0}
\renewcommand{\thesection}{S\arabic{section}}
\setcounter{figure}{0}
\renewcommand{\thefigure}{S\arabic{figure}}
\setcounter{theorem}{0}
\renewcommand{\thetheorem}{S\arabic{theorem}}
\setcounter{lemma}{0}
\renewcommand{\thelemma}{S\arabic{lemma}}
\setcounter{equation}{0}
\renewcommand{\theequation}{S\arabic{equation}}
\pagenumbering{roman}

\thispagestyle{empty}
\begin{center}
\textbf{\large Supplemental Material for ``Transition of Anticoncentration in Gaussian Boson Sampling''}\\
\vspace{2ex}
Adam Ehrenberg, Joseph T. Iosue, Abhinav Deshpande, Dominik Hangleiter, Alexey V. Gorshkov
\vspace{2ex}
\end{center}

In this Supplemental Material, we provide details behind many of the expressions in the main text. In particular, we derive expressions for the first and second moments of the output probabilities in the form of Kronecker $\delta$s and derive \cref{eqn:first_moment_delta,eqn:second_moment_intermediate,eqn:second_moment_polynomial} of the main text. In addition, we prove \cref{thm:first_moment}, \cref{thm:second_moment}, and \cref{lemma:k} that are presented as building blocks toward the proof of a transition in anticoncentration. We also explain more thoroughly the connection between the hiding property and the first moment of hafnians of generalized COE matrices. We further contextualize our definition of anticoncentration with respect to the literature and show how anticoncentration of the approximate distribution of output probabilities connects to anticoncentration of the true distribution. Finally, we use Scattershot Boson Sampling as intuition for why the transition in anticoncentration in Gaussian Boson Sampling exists.
\tableofcontents

\section{Algebraic Details of the First Moment---Derivation of Eq.~(5)}

In this section, we derive \cref{eqn:first_moment_delta} of the main text, which gives an expression for the first moment of the output probabilities in terms of Kronecker $\delta$s. We use \cref{eqn:hafnian} to expand the hafnian. We then use properties of independent Gaussians to simplify the expression. Specifically,
\begin{align}
    \underset{X\sim \mathcal{G}^{k \times 2n}}{\E}\bkt*{\abs{\haf(X^{\top}X)}^{2}} &= \paren*{\frac{1}{2^{n}n!}}^{2} \sum_{\sigma, \tau \in S_{2n}}\underset{X\sim \mathcal{G}^{k \times 2n}}{\E}\bkt*{\prod_{j=1}^{n}\paren*{\sum_{\ell_{j}=1}^{k}X_{\ell_{j}\sigma(2j-1)}X_{\ell_{j}\sigma(2j)}}\paren*{\sum_{o_{j}=1}^{k}X^{*}_{o_{j}\tau(2j-1)}X^{*}_{o_{j}\tau(2j)}}}\\
    &= \paren*{\frac{1}{2^{n}n!}}^{2} \sum_{\sigma, \tau \in S_{2n}}\underset{X\sim \mathcal{G}^{k \times 2n}}{\E}\bkt*{\sum_{\{\ell_{i},o_{i}\}_{i=1}^{n}=1 }^{k}\paren*{\prod_{j=1}^{n}X_{\ell_{j}\sigma(2j-1)}X_{\ell_{j}\sigma(2j)}X^{*}_{o_{j}\tau(2j-1)}X^{*}_{o_{j}\tau(2j)}}}\\
    &= \paren*{\frac{1}{2^{n}n!}}^{2} \sum_{\sigma, \tau \in S_{2n}}\sum_{\{\ell_{i},o_{i}\}_{i=1}^{n}=1}^{k}\underset{X\sim \mathcal{G}^{k \times 2n}}{\E}\bkt*{\paren*{\prod_{j=1}^{n}X_{\ell_{j}\sigma(2j-1)}X_{\ell_{j}\sigma(2j)}X^{*}_{o_{j}\tau(2j-1)}X^{*}_{o_{j}\tau(2j)}}}\\
    &= \paren*{\frac{1}{2^{n}n!}}^{2} \sum_{\sigma, \tau \in S_{2n}}\sum_{\{\ell_{i},o_{i}\}_{i=1}^{n}=1}^{k}\paren*{\prod_{j=1}^{n}\delta_{\ell_{j}o_{j'}}\delta_{\ell_{j}o_{j''}}},
\end{align}
where we have defined $j'$ to be the index such that $\sigma(2j-1) = \tau(2j'-1)$ or $\tau(2j')$. Similarly, $j''$ is the index such that $\sigma(2j) = \tau(2j''-1)$ or $\tau(2j'')$. Observe that $j' = j''$ if $\{\sigma(2j-1)$, $\sigma(2j)\} = \{\tau(2j-1)$, $\tau(2j)\}$ (note that this is an equality of \emph{sets}, meaning order does not matter). The first equation uses the definition of the hafnian, while the second follows from exchanging product and sum. The penultimate equation comes from the linearity of expectation. To get to the final equation, first recall that the $X_{ij}$ are i.i.d.~complex Gaussian random variables with mean $0$ and variance $1$. This means that the expectation value of a product of entries vanishes unless there are an equal number of unconjugated and conjugated copies of all indices. By the definition of $j'$, we ensure that the entry $X_{\ell_{j}\sigma(2j-1)}$ is matched to one of the $X^{*}$ entries as long as $\ell_{j}$ and $o_{j'}$ match, hence the first Kronecker $\delta$. The second Kronecker $\delta$ follows similarly.

We can exactly calculate $j'$:
\begin{equation}
    \sigma(2j-1) \in \{\tau(2j'-1), \tau(2j')\}
     \iff
     \tau^{-1}(\sigma(2j-1)) \in \{ 2j'-1, 2j'\}
     \iff
     \frac{\tau^{-1}(\sigma(2j-1))}{2} \in \{j' - \frac{1}{2}, j'\}.
\end{equation}
Thus:
\begin{equation}\label{eqn:j'_from_j}
    j' = \ceil*{\frac{\tau^{-1}(\sigma(2j-1))}{2}}.
\end{equation}
\noindent Similarly:
\begin{equation}\label{eqn:j''_from_j}
     j'' = \ceil*{\frac{\tau^{-1}(\sigma(2j))}{2}}.
\end{equation}
Therefore,
\begin{align}
     \underset{X\sim \mathcal{G}^{k \times 2n}}{\E}\bkt*{\abs{\haf(X^{\top}X)}^{2}} &= \paren*{\frac{1}{2^{n}n!}}^{2} \sum_{\sigma, \tau \in S_{2n}}\sum_{\{\ell_{i},o_{i}\}_{i=1}^{n}=1 }^{k}\paren*{\prod_{j=1}^{n}\delta_{\ell_{j}o_{\ceil*{\frac{\tau^{-1}(\sigma(2j-1))}{2}}}}\delta_{\ell_{j}o_{\ceil*{\frac{\tau^{-1}(\sigma(2j))}{2}}}}}\\
     &=\frac{(2n)!}{(2^{n}n!)^{2}}\sum_{\tau \in S_{2n}}\bkt*{\sum_{\{\ell_{i},o_{i}\}_{i=1}^{n}=1 }^{k}\paren*{\prod_{j=1}^{n}\delta_{\ell_{j}o_{\ceil*{\frac{\tau^{-1}(2j-1)}{2}}}}\delta_{\ell_{j}o_{\ceil*{\frac{\tau^{-1}(2j)}{2}}}}}}\\
     & = \frac{(2n)!}{(2^{n}n!)^{2}}\sum_{\tau \in S_{2n}}\bkt*{\sum_{\{o_{i}\}_{i=1}^{n}=1 }^{k}\paren*{\prod_{j=1}^{n}\delta_{o_{\ceil*{\frac{\tau(2j-1)}{2}}},o_{\ceil*{\frac{\tau(2j)}{2}}}}}}.\label{eqn:first_moment_delta_SM}
\end{align}
In the first equality, we have used \cref{eqn:j'_from_j,eqn:j''_from_j}.  In the second, we notice that $\tau$ and $\sigma$ occur only together as $\tau^{-1}\circ\sigma$, meaning we can perform a change of variables to convert our double summation over permutations in $S_{2n}$ to a single summation over a redefined $\tau^{-1}$ while gaining a factor $(2n)!$. The third equality comes from summing over the $\ell_{j}$ indices and redefining $\tau^{-1}\to\tau$. This is \cref{eqn:first_moment_delta} in the main text. 

\section{Algebraic Details of the Second Moment---Derivation of Eqs.~(8) and (9)}
In this section, we generalize the calculation of the previous section to the second moment of the output probabilities. The structure of the derivation is very similar, but the details are more nuanced due to the increased number of copies of $X$. 

We again begin with some algebraic manipulations:
\begin{align}
    \begin{split}
     &\underset{X\sim \mathcal{G}^{k \times 2n}}{\E}\bkt*{\abs{\haf(X^{\top}X)}^{4}} = \paren*{\frac{1}{2^{n}n!}}^{4} \sum_{\sigma,\tau,\alpha,\beta \in S_{2n}}\underset{X\sim \mathcal{G}^{k \times 2n}}{\E}\Biggr[\prod_{j=1}^{n}\paren*{\sum_{\ell_{j}=1}^{k}X_{\ell_{j}\sigma(2j-1)}X_{\ell_{j}\sigma(2j)}}\paren*{\sum_{o_{j}=1}^{k}X^{*}_{o_{j}\tau(2j-1)}X^{*}_{o_{j}\tau(2j)}}\\
    &\phantom{aaaaaaaaaa} \times \paren*{\sum_{p_{j}=1}^{k}X_{p_{j}\alpha(2j-1)}X_{p_{j}\alpha(2j)}}\paren*{\sum_{q_{j}=1}^{k}X^{*}_{q_{j}\beta(2j-1)}X^{*}_{q_{j}\beta(2j)}} \Biggr] 
    \end{split} \\
    \begin{split}
    &= \frac{1}{(2^{n}n!)^{4}} \sum_{\sigma,\tau,\alpha,\beta \in S_{2n}}\sum_{\{\ell_{i},o_{i},p_{i},q_{i}\}_{i=1}^{n}=1}^{k} \\
    &\phantom{aaaaaaaaaa}\underset{X\sim \mathcal{G}^{k \times 2n}}{\E}\Biggr[\prod_{j=1}^{n}X_{\ell_{j}\sigma(2j-1)}X_{\ell_{j}\sigma(2j)}X^{*}_{o_{j}\tau(2j-1)}X^{*}_{o_{j}\tau(2j)} X_{p_{j}\alpha(2j-1)}X_{p_{j}\alpha(2j)}X^{*}_{q_{j}\beta(2j-1)}X^{*}_{q_{j}\beta(2j)}\Biggr].\label{eqn:large_expectation}
    \end{split}
\end{align}
This first equation simply comes from the definition of the hafnian, and the second from exchanging product and sum and using the linearity of expectation. As in the proof of the first moment, we must properly match the indices of the Gaussian elements. Recall that, in order for the expectation value not to vanish, the indices $i,j$ must show up an equal number of times in a conjugated and non-conjugated copy of $X$ (otherwise, the expectation value of that term will vanish because our Gaussian is complex with zero mean). To proceed, first recall that permutations are bijective. Therefore, for all $j$ and any given permutation $\eta$, there is a unique value $y_{j}$ such that $\sigma(2j-1) = \eta(2y_{j}-1)$ or $\sigma(2j-1) = \eta(2y_{j})$. Similarly, there is a unique value $y_{j}'$ such that $\sigma(2j) = \eta(2y_{j}'-1)$ or $\sigma(2j) = \eta(2y_{j}')$. Using this bijectivity and the independence of matrix elements allows us to separate the single expectation value on the $8n$ matrix elements in \cref{eqn:large_expectation} into a product of $2n$ expectation values of $4$ elements:
\begin{equation}
\prod_{j=1}^{n}\underset{X\sim \mathcal{G}^{k \times 2n}}{\E}\bkt*{X_{\ell_{j}\sigma(2j-1)}X_{p_{k_{j}}\sigma(2j-1)} X^{*}_{o_{i_{j}}\sigma(2j-1)}X^{*}_{q_{m_{j}}\sigma(2j-1)}}\underset{X\sim \mathcal{G}^{k \times 2n}}{\E}\bkt*{X_{\ell_{j}\sigma(2j)}X_{p_{k_{j}'}\sigma(2j)} X^{*}_{o_{i_{j}'}\sigma(2j)}X^{*}_{q_{m_{j}'}\sigma(2j)}}.
\end{equation}
To explain more thoroughly: we have defined $i_{j}, k_{j}, m_{j}$ to be the indices that map to $\sigma(2j-1)$ under $\tau, \alpha, \beta$, respectively, in the sense that either $\eta(2y_{j}-1) = \sigma(2j-1)$ or $\eta(2y_{j}) = \sigma(2j-1)$ for $\eta \in \{\tau,\alpha,\beta\}$ and $y\in\{i, k, m\}$, respectively. Because two matrix elements are necessarily independent if they do not match on the second index, we can separate all elements with $\sigma(2j-1)$ as the second element into a single expectation value, hence the first term. To get the second term, we repeat this argument where $i_{j}', k'_{j}, m_{j}'$ are the indices that map to $\sigma(2j)$ under $\tau, \alpha, \beta$, respectively, in the sense that either $\eta(2y_{j}-1)=\sigma(2j)$ or $\eta(2y_{j}) = \sigma(2j)$ for $\eta \in \{\tau,\alpha,\beta\}$ and $y\in\{i, k, m\}$, respectively.

Now consider the first expectation value. For a nonvanishing expectation value, we must appropriately match the first indices of the matrix elements. We have three options: either all four indices can match, or the indices can be paired off in one of two ways. In the former case, the expectation value yields $2$ given that the elements are complex Gaussian with mean $0$ and variance $1$. By the same logic, the latter two cases yield an expectation value of $1$. In summary, 
\begin{align}
    \ell_{j} = p_{k_{j}} = o_{i_{j}} = q_{m_{j}} &\implies \mathbb{E} \to 2, \\
    (\ell_{j} \neq p_{k_{j}}) \wedge (\ell_{j} = o_{i_{j}}) \wedge (p_{k_{j}} = q_{m_{j}}) &\implies \mathbb{E} \to  1, \\
    (\ell_{j} \neq p_{k_{j}}) \wedge (\ell_{j} = q_{m_{j}}) \wedge (p_{k_{j}} = o_{i_{j}}) &\implies \mathbb{E} \to  1.
\end{align}
One might naively think that there should be another contribution from matching indices as
\begin{equation}
    (\ell_{j} = p_{k_{j}}) \wedge (\ell_{j} \neq q_{m_{j}}) \wedge (q_{m_{j}} = o_{i_{j}}). 
\end{equation}
However, the expectation value in this case actually vanishes, as we are working with \emph{complex} Gaussian random variables, meaning the indices need to be matched such that there are an equal number of conjugated and non-conjugated indices. 

We can write this in one simple expression using Kronecker $\delta$s as
\begin{equation}
    2\delta_{\ell_{j}p_{k_{j}}o_{i_{j}}q_{m_{j}}} + \delta_{\ell_{j}o_{i_{j}}}\delta_{p_{k_{j}}q_{m_{j}}}(1-\delta_{\ell_{j}p_{k_{j}}}) + \delta_{\ell_{j}q_{m_{j}}}\delta_{p_{k_{j}}o_{i_{j}}}(1-\delta_{\ell_{j}p_{k_{j}}}) =  \delta_{\ell_{j}o_{i_{j}}}\delta_{p_{k_{j}}q_{m_{j}}}+\delta_{\ell_{j}q_{m_{j}}}\delta_{p_{k_{j}}o_{i_{j}}}.
\end{equation}
That is,
\begin{equation}
     \underset{X\sim \mathcal{G}^{k \times 2n}}{\E}\bkt*{X_{\ell_{j}\sigma(2j-1)}X_{p_{k_{j}}\sigma(2j-1)} X^{*}_{o_{i_{j}}\sigma(2j-1)}X^{*}_{q_{m_{j}}\sigma(2j-1)}} =  \delta_{\ell_{j}o_{i_{j}}}\delta_{p_{k_{j}}q_{m_{j}}}+\delta_{\ell_{j}q_{m_{j}}}\delta_{p_{k_{j}}o_{i_{j}}},
\end{equation}
which is essentially an application of Isserlis'/Wick's theorem. 
Equivalent calculations as those used to derive \cref{eqn:j'_from_j,eqn:j''_from_j} can be made to rewrite each of $o_{i_{j}}, p_{k_{j}}, q_{m_{j}}$ in terms of $j$, giving
\begin{multline}
    \delta_{\ell_{j}o_{i_{j}}}\delta_{p_{k_{j}}q_{m_{j}}}+\delta_{\ell_{j}q_{m_{j}}}\delta_{p_{k_{j}}o_{i_{j}}} = \\
    \delta_{\ell_{j}o_{\ceil*{\frac{\tau^{-1}(\sigma(2j-1))}{2}}}}\delta_{p_{\ceil*{\frac{\alpha^{-1}(\sigma(2j-1))}{2}}}q_{\ceil*{\frac{\beta^{-1}(\sigma(2j-1))}{2}}}}+\delta_{\ell_{j}q_{\ceil*{\frac{\beta^{-1}(\sigma(2j-1))}{2}}}}\delta_{p_{\ceil*{\frac{\alpha^{-1}(\sigma(2j-1))}{2}}}o_{\ceil*{\frac{\tau^{-1}(\sigma(2j-1))}{2}}}}.
\end{multline}
Thus
\begin{multline}
    \underset{X\sim \mathcal{G}^{k \times 2n}}{\E}\bkt*{X_{\ell_{j}\sigma(2j-1)}X_{p_{k_{j}}\sigma(2j-1)} X^{*}_{o_{i_{j}}\sigma(2j-1)}X^{*}_{q_{m_{j}}\sigma(2j-1)}}\underset{X\sim \mathcal{G}^{k \times 2n}}{\E}\bkt*{X_{\ell_{j}\sigma(2j)}X_{p_{k'_{j}}\sigma(2j)} X^{*}_{o_{i'_{j}}\sigma(2j)}X^{*}_{q_{m'_{j}}\sigma(2j)}} = \\
    \paren*{\delta_{\ell_{j}o_{\ceil*{\frac{\tau^{-1}(\sigma(2j-1))}{2}}}}\delta_{p_{\ceil*{\frac{\alpha^{-1}(\sigma(2j-1))}{2}}}q_{\ceil*{\frac{\beta^{-1}(\sigma(2j-1))}{2}}}}+\delta_{\ell_{j}q_{\ceil*{\frac{\beta^{-1}(\sigma(2j-1))}{2}}}}\delta_{p_{\ceil*{\frac{\alpha^{-1}(\sigma(2j-1))}{2}}}o_{\ceil*{\frac{\tau^{-1}(\sigma(2j-1))}{2}}}}}  \\
    \times \paren*{\delta_{\ell_{j}o_{\ceil*{\frac{\tau^{-1}(\sigma(2j))}{2}}}}\delta_{p_{\ceil*{\frac{\alpha^{-1}(\sigma(2j))}{2}}}q_{\ceil*{\frac{\beta^{-1}(\sigma(2j))}{2}}}}+\delta_{\ell_{j}q_{\ceil*{\frac{\beta^{-1}(\sigma(2j))}{2}}}}\delta_{p_{\ceil*{\frac{\alpha^{-1}(\sigma(2j))}{2}}}o_{\ceil*{\frac{\tau^{-1}(\sigma(2j))}{2}}}}}.
\end{multline}
Therefore,
\begin{multline}
    \underset{X\sim \mathcal{G}^{k \times 2n}}{\E}\bkt*{\abs{\haf(X^{\top}X)}^{4}}= \paren*{\frac{1}{2^{n}n!}}^{4} \sum_{\sigma,\tau,\alpha,\beta \in S_{2n}}\sum_{\{\ell_{i},o_{i},p_{i},q_{i}\}_{i=1}^{n}=1}^{k}\Biggr[\prod_{j=1}^{n} \\
    \paren*{\delta_{\ell_{j}o_{\ceil*{\frac{\tau^{-1}(\sigma(2j-1))}{2}}}}\delta_{p_{\ceil*{\frac{\alpha^{-1}(\sigma(2j-1))}{2}}}q_{\ceil*{\frac{\beta^{-1}(\sigma(2j-1))}{2}}}}+\delta_{\ell_{j}q_{\ceil*{\frac{\beta^{-1}(\sigma(2j-1))}{2}}}}\delta_{p_{\ceil*{\frac{\alpha^{-1}(\sigma(2j-1))}{2}}}o_{\ceil*{\frac{\tau^{-1}(\sigma(2j-1))}{2}}}}}  \\
    \times \paren*{\delta_{\ell_{j}o_{\ceil*{\frac{\tau^{-1}(\sigma(2j))}{2}}}}\delta_{p_{\ceil*{\frac{\alpha^{-1}(\sigma(2j))}{2}}}q_{\ceil*{\frac{\beta^{-1}(\sigma(2j))}{2}}}}+\delta_{\ell_{j}q_{\ceil*{\frac{\beta^{-1}(\sigma(2j))}{2}}}}\delta_{p_{\ceil*{\frac{\alpha^{-1}(\sigma(2j))}{2}}}o_{\ceil*{\frac{\tau^{-1}(\sigma(2j))}{2}}}}}
    \Biggr].
\end{multline}
We can again reparameterize our sums over the permutations by performing a change of variables $ (\eta^{-1}\circ\sigma) \to \eta$ for $\eta \in \{\tau, \alpha, \beta\}$. This yields 
\begin{multline}
     \underset{X\sim \mathcal{G}^{k \times 2n}}{\E}\bkt*{\abs{\haf(X^{\top}X)}^{4}} = \paren*{\frac{1}{2^{n}n!}}^{4} (2n)! \sum_{\tau,\alpha,\beta \in S_{2n}}\sum_{\{\ell_{i},o_{i},p_{i}, q_i\}_{i=1}^{n}=1}^{k}\Biggr[\prod_{j=1}^{n} \\
    \paren*{\delta_{\ell_{j}o_{\ceil*{\frac{\tau(2j-1)}{2}}}}\delta_{p_{\ceil*{\frac{\alpha(2j-1)}{2}}}q_{\ceil*{\frac{\beta(2j-1)}{2}}}}+\delta_{\ell_{j}q_{\ceil*{\frac{\beta(2j-1)}{2}}}}\delta_{p_{\ceil*{\frac{\alpha(2j-1)}{2}}}o_{\ceil*{\frac{\tau(2j-1)}{2}}}}} \\ 
    \times \paren*{\delta_{\ell_{j}o_{\ceil*{\frac{\tau(2j)}{2}}}}\delta_{p_{\ceil*{\frac{\alpha(2j)}{2}}}q_{\ceil*{\frac{\beta(2j)}{2}}}}+\delta_{\ell_{j}q_{\ceil*{\frac{\beta(2j)}{2}}}}\delta_{p_{\ceil*{\frac{\alpha(2j)}{2}}}o_{\ceil*{\frac{\tau(2j)}{2}}}}}
    \Biggr].
\end{multline}
Expanding the product and summing over $\ell_{j}$ yields
\begin{multline}\label{eqn:second_moment_delta}
    \underset{X\sim \mathcal{G}^{k \times 2n}}{\E}\bkt*{\abs{\haf(X^{\top}X)}^{4}} = \paren*{\frac{1}{2^{n}n!}}^{4} (2n)! \sum_{\tau,\alpha,\beta \in S_{2n}}\sum_{\{o_{i},p_{i},q_{i}\}_{i=1}^{n}=1}^{k}\Biggr[\prod_{j=1}^{n} \\
    \Biggr(\delta_{o_{\ceil*{\frac{\tau(2j-1)}{2}}}o_{\ceil*{\frac{\tau(2j)}{2}}}}\delta_{p_{\ceil*{\frac{\alpha(2j-1)}{2}}}q_{\ceil*{\frac{\beta(2j-1)}{2}}}}\delta_{p_{\ceil*{\frac{\alpha(2j)}{2}}}q_{\ceil*{\frac{\beta(2j)}{2}}}} + \delta_{o_{\ceil*{\frac{\tau(2j-1)}{2}}}q_{\ceil*{\frac{\beta(2j)}{2}}}}\delta_{p_{\ceil*{\frac{\alpha(2j-1)}{2}}}q_{\ceil*{\frac{\beta(2j-1)}{2}}}}\delta_{p_{\ceil*{\frac{\alpha(2j)}{2}}}o_{\ceil*{\frac{\tau(2j)}{2}}}}
    +\\
   \delta_{q_{\ceil*{\frac{\beta(2j-1)}{2}}}o_{\ceil*{\frac{\tau(2j)}{2}}}}\delta_{p_{\ceil*{\frac{\alpha(2j-1)}{2}}}o_{\ceil*{\frac{\tau(2j-1)}{2}}}}\delta_{p_{\ceil*{\frac{\alpha(2j)}{2}}}q_{\ceil*{\frac{\beta(2j)}{2}}}} + \delta_{q_{\ceil*{\frac{\beta(2j-1)}{2}}}q_{\ceil*{\frac{\beta(2j)}{2}}}}\delta_{p_{\ceil*{\frac{\alpha(2j-1)}{2}}}o_{\ceil*{\frac{\tau(2j-1)}{2}}}}\delta_{p_{\ceil*{\frac{\alpha(2j)}{2}}}o_{\ceil*{\frac{\tau(2j)}{2}}}}\Biggr)
    \Biggr].
\end{multline}
This equation is the starting point of a new graph-theoretic approach.

As discussed in the main text, we use \cref{eqn:second_moment_delta} to define graphs, examples of which are provided in \cref{fig:second_moment_red_edges}(b) and \cref{fig:limits_examples}(a). Specifically, we let $G_{\tau,\alpha,\beta}(z)$ be a graph on $6n$ vertices, with labels $\{O_{i},P_{i},Q_{i}\}_{i=1}^{2n}$, and $z$ an integer from 1 to $4^{n}$. As was the case for the proof of the first moment, we use the Kronecker $\delta$s to define black and red edges. $z$ enumerates the different patterns of black edges, and $\tau, \alpha, \beta$ determine the red edges. Specifically, there is a red edge between $O_{j}$ and $O_{j'}$ if $\ceil{\tau(j)/2}=\ceil{\tau(j')/2}$, and similarly for the $O$ and $Q$ vertices using permutations $\alpha$ and $\beta$, respectively. However, given a choice of permutations, there are $4^{n}$ possible sets of black edges that correspond to the $4^{n}$ possible combinations of terms in \cref{eqn:second_moment_delta}. The sets of edges corresponding to each term are listed below: 
\begin{align}
   \delta_{o_{\ceil*{\frac{\tau(2j-1)}{2}}}o_{\ceil*{\frac{\tau(2j)}{2}}}}\delta_{p_{\ceil*{\frac{\alpha(2j-1)}{2}}}q_{\ceil*{\frac{\beta(2j-1)}{2}}}}\delta_{p_{\ceil*{\frac{\alpha(2j)}{2}}}q_{\ceil*{\frac{\beta(2j)}{2}}}} & \to \{(O_{2j-1}, O_{2j}), (P_{2j-1},Q_{2j-1}),(P_{2j},Q_{2j})\}, \\
    \delta_{o_{\ceil*{\frac{\tau(2j-1)}{2}}}q_{\ceil*{\frac{\beta(2j)}{2}}}}\delta_{p_{\ceil*{\frac{\alpha(2j-1)}{2}}}q_{\ceil*{\frac{\beta(2j-1)}{2}}}}\delta_{p_{\ceil*{\frac{\alpha(2j)}{2}}}o_{\ceil*{\frac{\tau(2j)}{2}}}} & \to \{(O_{2j-1}, Q_{2j}), (P_{2j-1},Q_{2j-1}),(O_{2j},P_{2j})\}, \\
     \delta_{q_{\ceil*{\frac{\beta(2j-1)}{2}}}o_{\ceil*{\frac{\tau(2j)}{2}}}}\delta_{p_{\ceil*{\frac{\alpha(2j-1)}{2}}}o_{\ceil*{\frac{\tau(2j-1)}{2}}}}\delta_{p_{\ceil*{\frac{\alpha(2j)}{2}}}q_{\ceil*{\frac{\beta(2j)}{2}}}} &\to  \{(O_{2j}, Q_{2j-1}), (P_{2j-1},O_{2j-1}),(P_{2j},Q_{2j})\}, \\
     \delta_{q_{\ceil*{\frac{\beta(2j-1)}{2}}}q_{\ceil*{\frac{\beta(2j)}{2}}}}\delta_{p_{\ceil*{\frac{\alpha(2j-1)}{2}}}o_{\ceil*{\frac{\tau(2j-1)}{2}}}}\delta_{p_{\ceil*{\frac{\alpha(2j)}{2}}}o_{\ceil*{\frac{\tau(2j)}{2}}}} &\to \{(O_{2j-1}, P_{2j-1}), (O_{2j},P_{2j}),(Q_{2j-1},Q_{2j})\}.
\end{align}
We refer to these sets of black edges as type-1, type-2, type-3, and type-4, respectively. We take the convention that our graphs have the vertices organized into three rows and $2n$ columns. The first, second, and third rows correspond to type-$O$, $-P$, and $-Q$ vertices, respectively. The columns are ordered by index $i$. Using this convention, black edges are constrained to lie within groups of two columns $2i-1$ and $2i$ using one of the four patterns described above. Again, see \cref{fig:second_moment_red_edges}(b) in the main text and \cref{fig:limits_examples}(a) later in this Supplemental Material for examples (please note that \cref{fig:limits_examples}(a) is not fully general, as it only has type-1 and type-4 black edges, but it does show that patterns of black edges can repeat, and it shows how $z$ identifies the patterns of black edges present in the graph). 

We repeat the conclusion of the main text, which is that we can map the number of ``free indices'' in the Kronecker $\delta$s to the number of connected components $C(G_{\tau,\alpha,\beta}(z))$ of the graph $G_{\tau,\alpha,\beta}(z)$. Each graph contributes $k^{C(G_{\tau,\alpha,\beta}(z))}$ to the sum, which means that the second moment can be written as
\begin{equation}\label{eqn:second_moment_intermediate_supp}
   M_2(k,n) = \frac{(2n)!}{\paren*{2^{n}n!}^{4}}\sum_{\tau,\alpha,\beta \in S_{2n}} \sum_{z \in [4^{n}]}k^{C(G_{\tau,\alpha,\beta}(z))},
\end{equation}
which is \cref{eqn:second_moment_intermediate} of the main text.
Removing the degeneracies induced by different permutations, and defining $\mathbb{G}^{2}_{n}(z)$ to be the set of graphs for the $z$th set of black edges and $\mathbb{G}^{2}_{n} \coloneqq \bigcup_{z=1}^{4^{n}}\mathbb{G}^{2}_{n}(z)$, we get a final result of  
\begin{equation}\label{eqn:second_moment_polynomial_supp}
 M_2(k,n) = (2n-1)!!\sum_{G\in\mathbb{G}^{2}_{n}}k^{C(G)}.
\end{equation}
This is \cref{eqn:second_moment_polynomial} of the main text.

\section{Proofs of Theorem~1, Theorem~2, and Lemma~1}
In this section, we give the proofs of \cref{thm:first_moment}, \cref{thm:second_moment}, and \cref{lemma:k} that were presented in the main text. We start with a restatement and proof of \cref{thm:first_moment}, which gives the first moment of the output probabilities. 
\begin{theorem} \label{thm:first_moment_supp}
The sum over graphs in $\mathbb{G}^{1}_{n}$ satisfies 
\begin{equation}\label{eqn:g_first_moment_supplement}
   \sum_{G\in\mathbb{G}^{1}_{n}}k^{C(G)} = k(k+2)\dots (k+2n-2).
\end{equation}
and hence $M_1(k,n) = (2n-1)!!(k+2n-2)!!/(k-2)!!$.
\end{theorem}
\begin{proof}
We proceed by induction on $n$. Let $f(k,n)$ be the LHS of \cref{eqn:g_first_moment_supplement}. For the base case $n=1$, there is only a single possible graph $G$ that has a single connected component. Thus $f(k,1) = k$. For the inductive step, which is visualized in \cref{fig:first_moment_induction}, consider two subsets of $\mathbb{G}_{n}^{1}$. The first set has graphs that possess a red edge between $O_{1}$ and $O_{2}$, which means that these two vertices form their own connected component (recall that $O_{1}$ and $O_{2}$ are always connected with a black edge). Summing $k^{C(G)}$ over all graphs of this type then yields a contribution of $k f(k,n-1)$. The other subset of $\mathbb{G}_{n}^{1}$ has graphs that possess a red edge between $O_{1}$ and a vertex besides $O_{2}$, say $O_{x}$. In these graphs, the number of connected components in the graph does not change if one collapses the three vertices $O_{1}$, $O_{2}$, and $O_{x}$ into a single vertex (because they are all connected by either a black or red edge). Therefore, because there are $2n-2$ choices for the vertex $O_{x}$ linked to $O_{1}$ by a red edge, we get an overall contribution of $(2n-2)f(k,n-1)$ when summing $k^{C(G)}$ over these graphs. 
\begin{figure}[ht!]
    \centering
    \includegraphics[width=.75\linewidth]{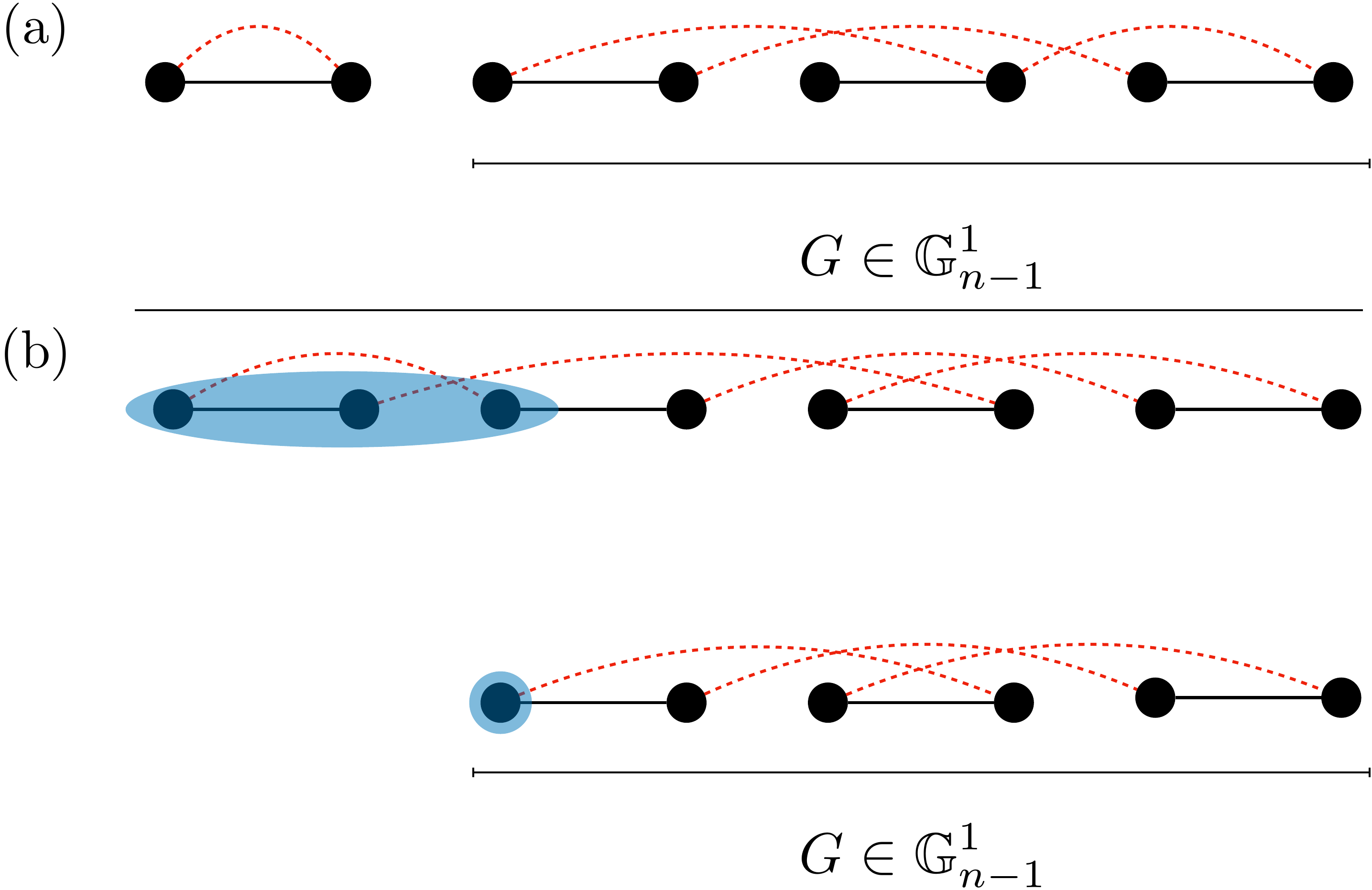}
    \caption{Visualization of the inductive step in the proof of the first moment of the output probabilities (\cref{thm:first_moment} in the main text, restated here in the SM as \cref{thm:first_moment_supp}). The inductive step proceeds in two cases that are determined by the red edge that connects to the first vertex in the graph in $\mathbb{G}_{n}^{1}$. 
    In (a), we consider the case where the first two vertices, which are linked by a black (solid) edge, are also linked by a red (dashed) edge, meaning they comprise a single connected component. This contributes a factor of $k$ times the contribution from a graph in $\mathbb{G}_{n-1}^{1}$, which comes from the remaining $2n-2$ vertices and their edges. (b) considers the case where the first vertex is linked via a red edge to a different vertex $O_{x} \neq O_{2}$ for which there are $2n-2$ choices (here $a=3$). The number of connected components does not change after identifying and combining the three vertices that are connected in this way (visualized by the blue background), meaning we again reduce down to a graph in $\mathbb{G}_{n-1}^{1}$, but this time without the multiplicative factor of $k$.}
    \label{fig:first_moment_induction}
\end{figure}

Overall then, we find that
\begin{align}
    f(k,n) &= kf(k, n-1) + (2n-2)f(k,n-1)\\
           &= (k+2n-2)f(k,n-1) \\
           &= (k+2n-2)(k+2n-4)\dots (k+2) k,
\end{align}
which proves the formula (and where the inductive hypothesis is used in the final equality).
\end{proof}

We note briefly that the structure of this proof is similar to that used in Ref.~\cite{arrazolaQuantumApproximateOptimization2018} to calculate $\underset{X\sim\mathcal{G}^{n\times n}}{\E}[\abs{\haf{X}}^{4}]$. The proofs are similar because there are four copies of $X$ in each, but the proofs are not identical given that different definitions of the hafnian are used. Additionally, similar graphs, and a similar calculation involving enumerating the number of graphs of a given number of connected components, show up in the bioinformatic study of breakpoint graphs (which are a type of graph defined by two perfect matchings that show up in the theory of comparative genomics) \cite{feijaoDistributionCyclesPaths2015}. 

We next move on to a proof of \cref{thm:second_moment}, which gives the form of the second moment as a polynomial in $k$. We again restate the theorem for convenience.

\begin{theorem}\label{thm:second_moment_supp}
The second moment $M_2(k,n)$ is a degree-$2n$ polynomial in $k$ and can be written as $ M_2(k,n) = (2n-1)!! \sum_{i = 1}^{2n} c_{i}k^{i}$, where $c_{i}$ is the number of graphs $G\in \mathbb{G}^{2}_{n}$ that have $i$ connected components. 
\end{theorem}
\begin{proof}
As mentioned in the main text, once \cref{eqn:second_moment_polynomial} is derived, the theorem follows after deriving the correct limits of summation. Trivially, the fewest possible number of connected components is $1$. To see that the largest possible number of connected components is $2n$, we consider the four patterns of black edges that are illustrated in \cref{fig:second_moment_red_edges}(b) of the main text and how many connected components can possibly occur in graphs with those different patterns. See also \cref{fig:limits_examples} for a reminder of the patterns of black edges and a visual explanation of the following argument. 

First note that, because all vertices are paired via black edges, every connected component has an even number of vertices. Therefore, the two smallest sizes of connected components are $2$ and $4$ vertices. In order to get a connected component of size $2$, one must connect a pair of vertices with both a black and a red edge. 
Red edges are constrained to lie in a single row, meaning only type-1 and type-4 patterns of black edges, which contain a pair of vertices connected by a black edge in the same row, can yield a connect component of size $2$.
Pairing off the remaining vertical black edges yields connected components of size $4$, the next smallest size.

Therefore, the maximum number of connected components arises from taking only type-1 and type-4 edges. This requires connecting each horizontal black edge by red edge (creating a connected component with $2$ vertices) and then pairing off the vertical edges coming from the same type. This allows for the maximal $2$ connected components per set of six vertices, meaning $2n$ total connected components. 
\end{proof}
\begin{figure}[ht!]
    \centering
    \includegraphics[width = 0.75\linewidth]{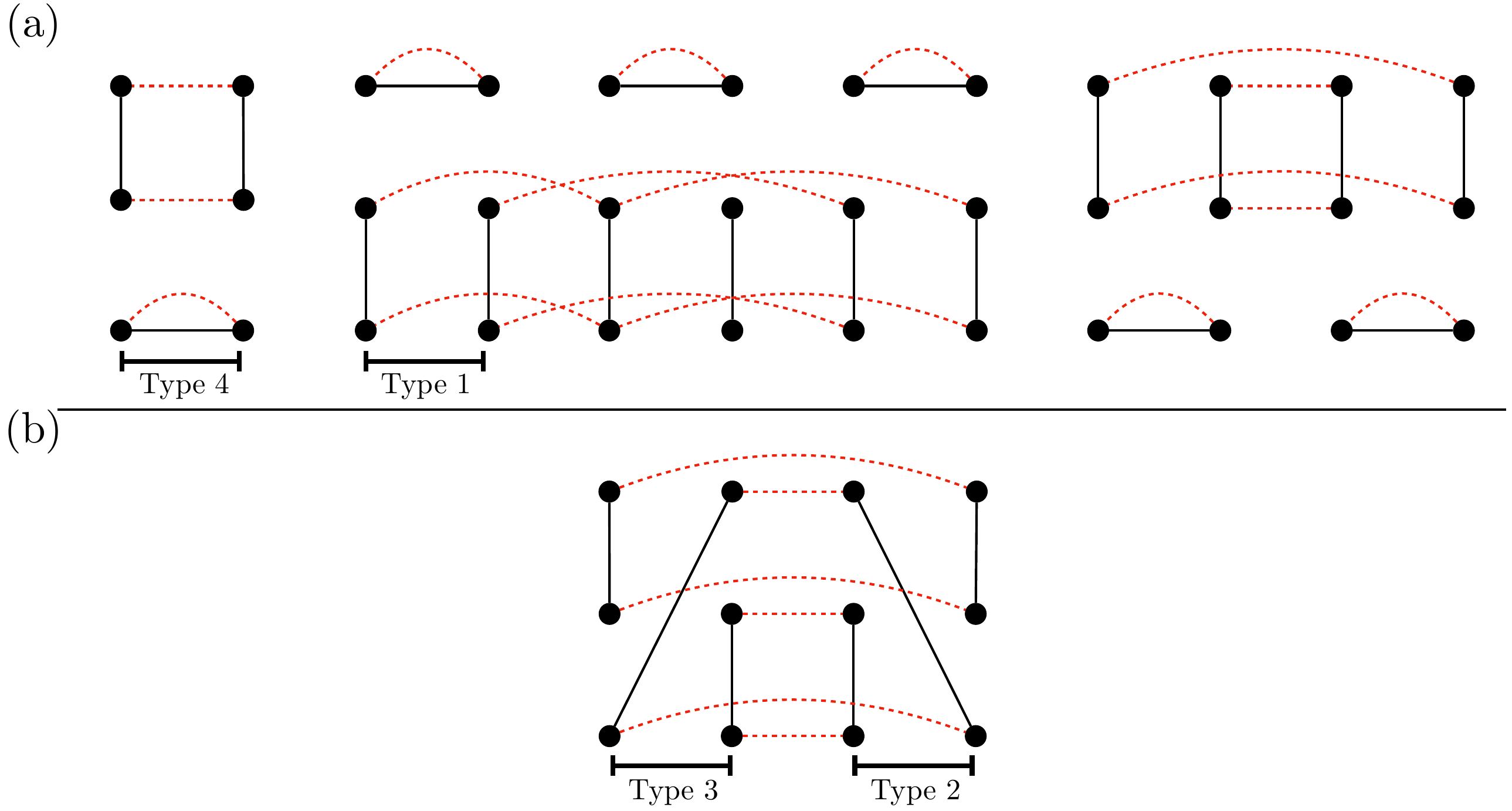}
    \caption{(a) Example graph in $\mathbb{G}_{6}^{2}$ showing how to achieve an average of two connected components per set of six vertices using only type-1 and type-4 sets of edges. All vertices connected by horizontal black (solid) edges are also connected by red (dashed) edges. All type-1 vertical edges are paired off, as are type-4 vertical edges. Note that this graph would correspond to $z =1 + 3\times 4^{5} + 0 \times 4^{4} + 0 \times 4^{3} + 0 \times 4^{2} + 3 \times 4^{1} + 3 \times 4^{0} = 3088$. (b) Example showing how using type-2 and type-3 black edges lead to, at most, three connected components per two sets of six vertices.}
    \label{fig:limits_examples}
\end{figure}

We also prove \cref{lemma:k}, which we again restate for convenience:

\begin{lemma}\label{lemma:k_supp}
We have that 
\begin{enumerate}
    \item[i.] $M_2(1,n) = ((2n-1)!!)^4 4^n$
    \item[ii.]  $ c_{2n} = (2n)!!$
\end{enumerate}
\end{lemma}

\begin{proof}
Part (i): examine \cref{eqn:second_moment_delta}. Because $k = 1$, $o_{i} = p_{i} = q_{i} = 1$ for all $i$. Thus, regardless of the permutation, all  Kronecker $\delta$s are always satisfied. This means that, independent of the permutation, each factor is always $4$ such that the product becomes $4^{n}$. The sum over the three copies of $S_{2n}$ then simply yields a factor of $(2n)!^{3}$. The result then follows. 

Part (ii): we argued in the proof of \cref{thm:second_moment_supp} that the leading-order term in the polynomial expansion of the second moment is $k^{2n}$, and it comes from graphs that consist of only type-1 and type-4 black edges. Each type-1 and type-4 set of edges contains a horizontal black edge, and the two vertices linked by that black edge also must be linked by a red edge to create a $2$-vertex connected component. Additionally, the vertical edges of the type-1 sets need to be paired off via red edges; similarly, the vertical edges of the type-4 sets need to be paired off. This ensures that each other connected component has exactly $4$ vertices, maximizing the number of possible connected components.  

\cref{fig:leading_order} visualizes how to now reduce the remaining calculation to the value of the first moment when $k = 2$. If we imagine collapsing each pair of adjacent vertical edges (i.e., those coming from the same group of $6$ vertices) onto a pair of vertices connected by a black edge, we reproduce the atomic graph from the proof of the first moment. Here, by atomic graph, we mean the vertices and the fixed black edges which are shared by all graphs; the red edges are not yet included. Explicitly, there are $2n$ vertices, and vertices $O_{2i-1}, O_{2i}$ are connected with a black edge. The black edges here act to identify that the original uncollapsed vertical edges were of the same type. Drawing red edges in the simplified graph on $2n$ vertices corresponds to pairing off vertical edges in the original graph on $6n$ vertices with red edges. Note that this also implies that red edges connect vertical edges of the same type. Therefore, a connected component in the simplified graph could correspond to two preimages in the original graph: either all type-1 vertical edges, or all type-4 vertical edges. Then, by summing over all graphs and weighting each connected component by $2$, we are effectively evaluating $f(2, n)$ in \cref{eqn:g_first_moment_supplement}, i.e.:
\begin{equation}
    f(2,n) = \frac{(k+2n-2)!!}{(k-2)!!}\bigg|_{k=2} = (2n)!!.
\end{equation}

\end{proof}
\begin{figure}[ht!]
    \centering
    \includegraphics[width=.6\linewidth]{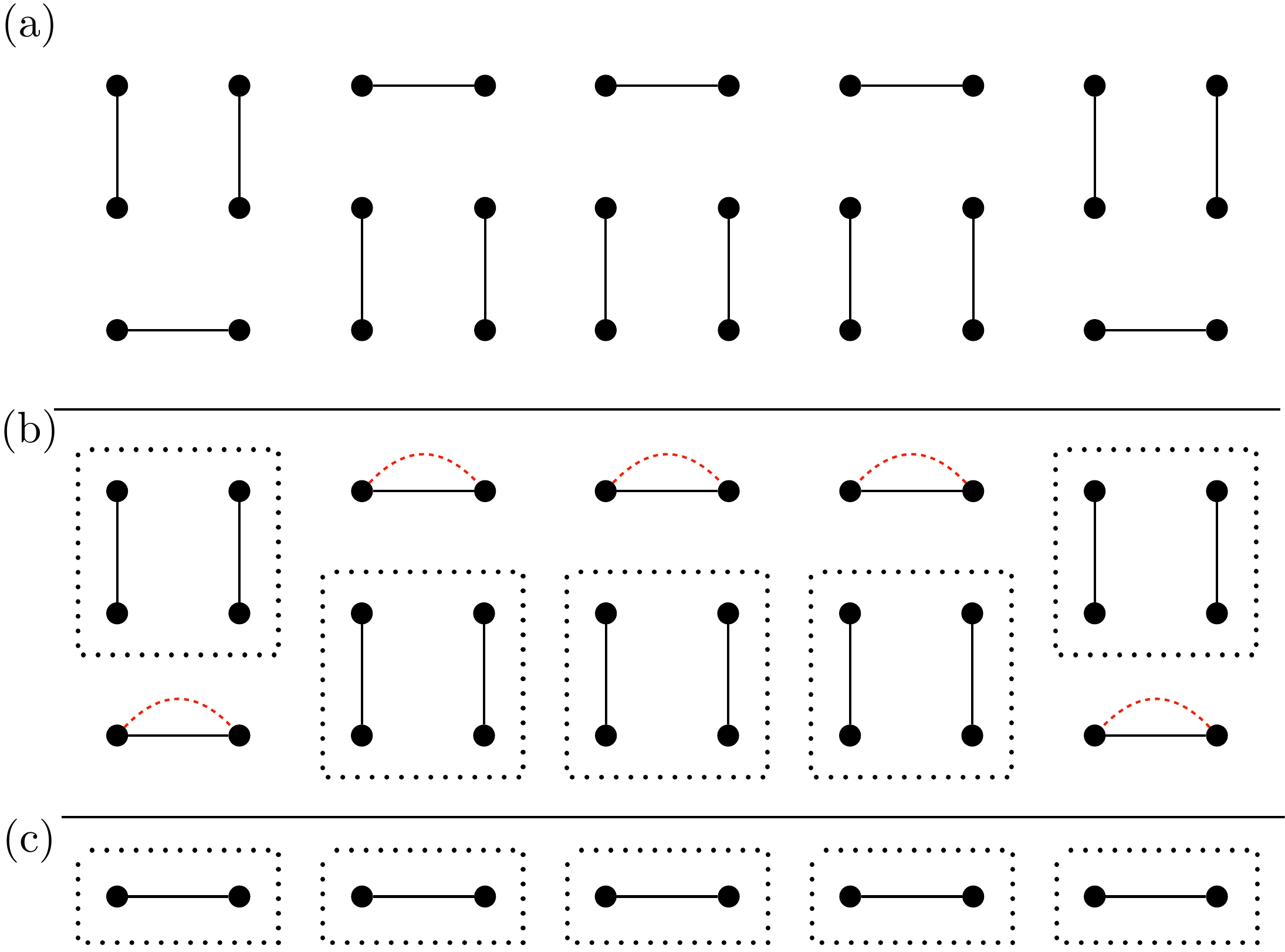}
    \caption{Visualization of how the calculation of the coefficient of the leading-order term in the second moment can be reduced to the $k=2$ case of the first moment.
    Recall that black edges are solid and red edges are dashed. (a) As proven in \cref{thm:second_moment}, graphs that maximize the number of connected components contain only type-1 and type-4 black edges. (b) To maximize the number of connected components, the horizontal black edges must form their own connected component with two vertices, meaning their vertices must be connected by a red edge. Furthermore, each vertical black edge must be paired off with exactly one other vertical black edge of the same type, forming a connected component with $4$ vertices. We draw dotted boxes around the two black vertical edges to show that they come from the same type. (c) If we collapse each vertical edge onto a single vertex and then connect that vertex to the vertex stemming from its adjacent edge in the original graph (i.e. the other vertical edge from the same group of six vertices), then we reduce to the atomic graph (i.e., the graph with the fixed black edges, but without red edges) from the proof of the first moment. Red edges on this collapsed graph would then correspond to pairing off vertical edges in the original graph with red edges. Because paired edges in the original graph can only exist between edges of the same type, each connected component in the simplified graph could have come from either type-1 or type-4 vertical edges. This is equivalent to setting $k$, the contribution from each connected component, to $2$, and then evaluating $f(2,n)$.}
    \label{fig:leading_order}
\end{figure}

\section{Approximate Hiding and Asymptotics of the First Moment}\label{sec:asymptotics}
In this section, we discuss more thoroughly the connection between hiding, the relevant sample space, and the first moment of squared hafnians of generalized COE matrices.  

In the main text, we introduce the normalized average outcome-collision probability as a measure of anticoncentration. Fixing the output state to have $2n$ photons, we write this as $|\Omega_{2n}| \E_{U \in U(m)} [\sum_{\vec n \in \Omega_{2n}} P_U(\vec n)^2] $, where $\Omega_{2n}$ is the space of collision-free outcomes with $2n$ photons in $m$ modes, and its size, which we write as  $|\Omega_{2n}|$, is simply $\binom{m}{2n}$. We here work specifically with the non-collisional sample space because, in order for hiding to hold, collisions have to be negligible (a non-negligible likelihood of repeated columns in $U^{\top}_{1_{k},\vec{n}}U_{1_{k},\vec{n}}$ would prevent this distribution from being well approximated by $X^{\top}X$ with $X$ Gaussian). And, indeed, when $n = o(\sqrt{m})$, it is easy to see that the size of the full sample space of $2n$ photons in $m$ modes, $\binom{m+2n-1}{2n}$, approaches $|\Omega_{2n}| = \binom{m}{2n}$ when $n \gg 1$. In particular, 
\begin{equation}
    \frac{m^{2n}}{(2n)!} \leq \frac{(m+2n-1)!}{(m-1)!(2n)!} \leq \frac{(m+2n-1)^{2n}}{(2n)!}
\end{equation}
and 
\begin{equation}
   \frac{\frac{(m+2n-1)^{2n}}{(2n)!}}{ \frac{m^{2n}}{(2n)!}} = \left(1+ \frac{2n-1}{m}\right)^{2n} \stackrel{n \gg 1}{\longrightarrow} 1
\end{equation}
because $(2n-1)/m = o(1/n)$. That is, $\Omega_{2n}$ is the dominant contribution to the full sample space. 

We proceed to then replace $|\Omega_{2n}|$ with the expected value of the outcome probabilties, $\mathbb{E}_{U}[P_{U}(\vec{n})]$, that is, the first moment over input unitaries of a specific outcome. This holds assuming that the hiding property in Conjecture~\ref{con:gaussian} holds. Roughly, hiding ensures that we do not preference any individual outcome, meaning we can replace the expected value over all probabilities with that over unitaries for a single probability. By linearity of expectation (and the fact that probabilities sum to unity), this expectation over unitaries should simply be the inverse of the size of the sample space of non-collisional outcomes. Finally, Conjecture~\ref{con:gaussian} also gives us an approximate equality between $\mathbb{E}_{U}[P_{U}(\vec{n})]$ and $M_{1}$, the first moment of the squared hafnian of generalized COE matrices (properly rescaled to contain the correct prefactors).

We therefore now show that our calculation of the first moment in the hiding regime is consistent with the above discussion in the sense that $\mathbb{E}_{U}[P_{U}(\vec{n} | \sum_i \vec{n}_i = 2n)] = \mathbb{E}_{U}[P_{U}(\vec{n})]/P(2n)$ is asymptotically equal to $|\Omega_{2n}|^{-1}= \binom{m}{2n}^{-1}$ assuming Conjecture~\ref{con:gaussian}. Here, $P(2n)$ is the probability that our output is in the $2n$-photon sector (i.e., the probability that $\Omega_{2n}$ is the proper sample space to consider in the first place). 

Recall our input state has the first $k$ of $m$ modes prepared in the single-mode squeezed vacuum state with identical squeezing parameter $r$, and the remaining $m - k$ modes are prepared in the vacuum state. 
The probability of an outcome $\vec{n}$ is given by \cref{eqn:gbs_probability}:
\begin{equation}
    P_{U}(\vec{n}) = \frac{\tanh^{2n} r}{\cosh^{k} r}\abs{\haf(U^{\top}_{1_{k},\vec{n}}U_{1_{k},\vec{n}})}^{2},
\end{equation}
where $U_{1_{k},\vec{n}}$ is the submatrix of $U$ given by the first $k$ rows and the columns dictated by where $\vec{n}$ is nonzero. Define $\tilde{U}_{1_{k},\vec{n}}\coloneqq mU_{1_{k},\vec{n}}$. Using multiplicativity of the Hafnian, one finds
\begin{equation}
      P_{U}(\vec{n}) = \frac{\tanh^{2n} r}{\cosh^{k}r}\frac{1}{m^{2n}}\abs{\haf\paren*{\tilde{U}_{1_{k},\vec{n}}\tilde{U}^{\top}_{1_{k},\vec{n}}}}^{2}.
\end{equation}
Assuming Conjecture~\ref{con:gaussian}, then $U^{\top}_{1_{k},\vec{n}}U_{1_{k},\vec{n}} \sim X^{\top}X$ where $X \sim \mathcal{N}(0,1/m)_{c}^{k \times 2n}$, which means $\tilde{U}^{\top}_{1_{k},\vec{n}}\tilde{U}_{1_{k},\vec{n}}\sim X^{\top}X$, but now $X\sim \mathcal{N}(0,1)_{c}^{k \times 2n}$. Then, by Conjecture~\ref{con:gaussian} and \cref{thm:second_moment}, we find
\begin{equation}
    \underset{U\in\mathrm{U}(m)}{\E}\bkt*{\abs{\haf\paren*{\tilde{U}^{\top}_{1_{k},\vec{n}}\tilde{U}_{1_{k},\vec{n}}}}^{2}} \approx  \underset{X\in\mathcal{G}^{k\times 2n}}{\E}\bkt*{\abs{\haf\paren*{X^{\top}X}}^{2}} = \frac{(2n)!}{2^{n}n!}\frac{(k+2n-2)!!}{(k-2)!!},
\end{equation}
where the first part of the equation is not an equality precisely because the hiding in Conjecture~\ref{con:gaussian} is not exact. This implies that
\begin{equation}\label{eqn:asymptotic_intermediate}
    \mathbb{E}_{U}[P_{U}(\vec{n})] \approx \frac{\tanh^{2n}r}{\cosh^{k}r}\frac{1}{m^{2n}}\frac{(2n)!}{2^{n}n!}\frac{(k+2n-2)!!}{(k-2)!!}.
\end{equation}

Now, a single-mode squeezed vacuum state with squeezing parameter $r$ and phase $\phi$ has Fock-state expansion given by
\begin{equation}
    \ket{\textrm{SMSV}} = \frac{1}{\sqrt{\cosh r}}\sum_{\ell = 0}^{\infty} (-e^{i\phi} \tanh r)^{\ell}\frac{\sqrt{(2\ell)!}}{2^{\ell}\ell!}\ket{2\ell}.
\end{equation}
Therefore, the probability of measuring $2\ell$ photons is 
\begin{equation}
    \abs{\braket{2\ell|\textrm{SMSV}}}^{2} = \frac{\tanh^{2\ell r}}{\cosh r}\frac{(2\ell)!}{(2^{\ell}\ell!)^{2}}. 
\end{equation}
Given $k$ independent single-mode squeezed vacuum states, the probability of finding $2n$ total photons is the $k$-fold convolution of the Fock-basis probability distribution of one single-mode squeezed vaccuum state:
\begin{equation}
    P(2n) = \sum_{2\ell_{1} + \dots + 2\ell_{k} = 2n}\prod_{i=1}^{k}\frac{\tanh^{2\ell_{i} r}}{\cosh r}\frac{(2\ell_{i})!}{(2^{\ell_{i}}\ell_{i}!)^{2}} = \frac{\tanh^{2n}r}{\cosh^{k}r}\frac{1}{2^{2n}} \sum_{2\ell_{1} + \dots + 2\ell_{k} = 2n}\prod_{i=1}^{k}\binom{2\ell_{1}}{\ell_{1}}. 
\end{equation}
This probability distribution is unchanged if the $k$ independent single-mode squeezed vacuum states are acted upon by a linear-optical unitary before measurement (such a unitary does not change the photon number, only the location of the photons). The combinatorial identity at the core of this $k$-fold convolution has been calculated before in Refs.~\cite{liMultifoldConvolutionsBinomial2019,changGeneralizationProbabilisticProof2011}. Specifically, 
\begin{equation}\label{eqn:k-fold_convolution}
    \sum_{2\ell_{1} + \dots + 2\ell_{k} = 2n}\prod_{i=1}^{k}\binom{2\ell_{1}}{\ell_{1}} = 4^{n}\binom{n-1+k/2}{n},
\end{equation}
where we note that \cref{eqn:k-fold_convolution} holds even in the case where $k$ is odd using a generalization of the binomial coefficients in terms of the $\Gamma$ function.

The overall probability of finding $2n$ photons from $k$ independent single-mode squeezed vacuum states, even after the application of a linear optical unitary, is therefore
\begin{equation}
    P(2n) = \frac{\tanh^{2n}r}{\cosh^{k}r}\frac{1}{2^{2n}} 4^{n}\binom{n-1+k/2}{n} = \frac{\tanh^{2n}r}{\cosh^{k}r}\binom{n-1+k/2}{n}. 
\end{equation}
We note that this expression, but not the full derivation, is also provided in Ref.~\cite{kruse_detailed_2019}. A bit of algebraic manipulation reveals
\begin{equation}
     P(2n) = \frac{\tanh^{2n}r}{\cosh^{k}r}\binom{n-1+k/2}{n} = \frac{\tanh^{2n}r}{\cosh^{k}r}\frac{(2n-1)!!(k+2n-2)!!}{(2n)!(k-2)!!} = \frac{1}{(2n)!} \frac{\tanh^{2n}r}{\cosh^{k}r} \underset{X\sim \mathcal{G}^{k \times 2n}}{\E}\bkt*{\abs{\haf(X^{\top}X)}^{2}}.
\end{equation}
According to \cref{eqn:asymptotic_intermediate}, then
\begin{equation}
    P(2n) \approx \frac{m^{2n}}{(2n)!}\mathbb{E}_{U}[P_{U}(\vec{n})],
\end{equation}
which finally implies
\begin{equation}
    \frac{\mathbb{E}_{U}[P_{U}(\vec{n})]}{P(2n)} \approx \frac{(2n)!}{m^{2n}} \approx \binom{m}{2n}^{-1} = |\Omega_{2n}|^{-1}, 
\end{equation}
where the first approximation is due to the fact that hiding is not exact, and the second approximation holds in the photon non-collisional regime. 

\section{Details on Definitions of Anticoncentration}
\label{app:definitions of anticoncentration}
In this section, we discuss some of the details behind our definition of anticoncentration and how it relates to the standard notion of anticoncentration often used in the literature. We also discuss how these different definitions interact when it comes to showing anticoncentration holds for the exact distribution of the output probabilities of GBS given anticoncentration of the approximate distribution. 
\subsection{Definitions}
We first discuss in somewhat more detail the relevance of anticoncentration to the argument for hardness of sampling from the output distribution of Gaussian Boson Sampling (GBS). This argument makes use of an approximate counting algorithm due to \textcite{stockmeyer_complexity_1983}. 
Roughly, we assume that there is an efficient sampling algorithm for GBS that, given a linear-optical unitary $U$, samples from a distribution $Q_U$ which is close up to a constant $\epsilon > 0$ to the ideal GBS distribution $P_U$ (recall \cref{eqn:gbs_probability} of the main text) in total-variation distance 
\begin{align}
    \mathsf{tvd}(P_U,Q_U) \coloneqq \frac 12 \sum_{\vec n} | P_U(\vec n) - Q_U(\vec n) |  \leq \epsilon. 
\end{align}
Supposing such a sampling algorithm exists, and given the so-called hiding property (see Sec. D of Ref.~\cite{hangleiterComputationalAdvantageQuantum2023} for details), we can use it as input to Stockmeyer's algorithm. 
Stockmeyer's algorithm then approximates the probability $P_U(\vec n)$ up to an error given by 
\begin{align}
\label{eq:error stockmeyer}
    \varepsilon = \frac 1 {\mathrm{poly}(n)} P_U(\vec n) + \frac{2 \epsilon}{|\Omega | \delta } \left( 1 + \frac 1 {\mathrm{poly}(n)} \right), 
\end{align}
with probability $1-\delta$ over $\vec n$, where $\Omega$ is the sample space on which $P_U$ is defined.
If it is sufficiently hard (\#\textsf{P}-hard, to be precise) to approximate the outcome probabilities $P_U(\vec n)$ up to the error \eqref{eq:error stockmeyer}, on the instances on which our approximation scheme achieves this error, this rules out the approximate sampling algorithm up to very reasonable complexity-theoretic conjectures (one of which is the non-collapse of the polynomial hierarchy, a generalization of the famous \textsf{P}\,$\neq$\,\textsf{NP}\ conjecture). 
The required property is thus what we call ``approximate average-case hardness,'' that is, the statement that any algorithm which is able to compute $P_U(\vec n)$ with probability $1-\delta$ over the instances up to the error \eqref{eq:error stockmeyer} is able to solve any \#\textsf{P}-hard problem (of the same difficulty as approximating the outcome probabilities $P_{U}(\vec{n})$ up to the error \eqref{eq:error stockmeyer}). 

While we know average case hardness of approximating the outcome probabilities up to error $2^{- \Omega(n \log n)}$ \cite{bouland_noise_2022,deshpandeQuantumComputationalAdvantage2022a}, it is only conjectured for the relevant approximation error given by either $c_1 P_U(\vec n)$ or $c_2/|\Omega|$ for constants $c_1, c_2>0$. 
Anticoncentration serves as evidence for the truth of the conjecture, the idea being the following: 
suppose that most of the outcome probabilities are very close to zero, i.e. $\ll \epsilon/2^{-n}$, meaning only a vanishing fraction of them are relevant. 
Then a high approximation error on the relevant probabilities is tolerable, because we only need to distinguish between relevant and irrelevant outcomes, and a sufficiently good approximation to the irrelevant ones is zero. 
This is a significantly easier task than if the distribution is highly spread out and a large fraction of the probabilities is ``relevant'' in the sense that all of the relevant probabilities are of the same order of magnitude as the uniform distribution. 

In the standard argument, this intuition is formalized as the statement   
\begin{align}
\label{eq:anticoncentration}
    \Pr_{U \in \mathrm U(m)}\left[ P_U(\vec n) \geq \frac \alpha {|\Omega|} \right]  \geq \gamma(\alpha),
\end{align}
for some constants $\alpha, \gamma(\alpha) > 0 $.
In this formulation, we have made crucial use of the hiding property, which asserts that the distribution over circuits is invariant under a procedure by which we ``hide'' a particular outcome $\vec n$ in the probability of obtaining a different outcome $\vec n'$ of a random circuit. 
This allows us to restrict our attention to the distribution over circuits of a fixed outcome $\vec n$.

The anticoncentration property \eqref{eq:anticoncentration} implies that the error \eqref{eq:error stockmeyer} is dominated by the first term on a $\gamma(\alpha)(1- \delta)$ fraction of the instances because with probability $\gamma(\alpha)$ we can upper bound the second term by $P_U(\vec n)$. 
But, if a large fraction of the probabilities is larger than uniform, then none of them can be much larger than uniform and, hence, the approximation error needs to be exponentially small. 
Thus, we expect that, in the presence of anticoncentration, approximating the outcome probabilities up to the error \eqref{eq:error stockmeyer} is much harder than without anticoncentration, lending credibility to the approximate average-case hardness conjecture. 

In our definition of anticoncentration, we consider the (normalized) average collision probability 
\begin{align}
\label{eq:average collision probability}
   P_2(\mathrm U(m)) \coloneqq |\Omega| \sum_{\vec n \in \Omega} \E_{U \in \mathrm U(m)} \left[P_U(\vec n)^2 \right] \\
   \stackrel{\text{hiding}}{=}|\Omega|^2 \E_{U \in \mathrm U(m)} \left[P_U(\vec n)^2 \right]. 
\end{align}
The collision probability is the probability that, were one to sample the distribution twice, one would receive the same outcome both times. For very flat distributions it is very small. 
With the normalization, the collision probability of the uniform distribution is given by $1$, which is its minimal value. 
On the other hand, the normalized collision probability of a fully peaked distribution with a single unit probability is given by $|\Omega|$. 

The average collision probability is thus another measure of the anticoncentration of the outcome probabilities in the ensemble of linear-optical unitaries. 
It is a more coarse-grained measure, though, because it is only an average quantity. 
Indeed, a (constantly) small average collision probability implies anticoncentration in the sense of \eqref{eq:anticoncentration} via the Paley-Zygmund inequality as 
\begin{align}
    \Pr_{U \in \mathrm U(m)} \left[ P_U(\vec n) \geq \frac \alpha {|\Omega|} \right] \geq (1- \alpha)^2 \frac 1 {P_2(\mathrm U(m))}. 
\end{align}
The relevant quantity of interest to anticoncentration is thus the inverse average collision probability $p_2(\mathrm U(m)) = 1/P_2(\mathrm U(m)$. 
Because by hiding the first moment $\E_U[P_U(\vec n)]$ must evaluate to the inverse size of the sample space, we can rewrite $p_2$ for GBS as 
\begin{align}
     p_2(\mathrm U(m)) = \frac{\E_{U \in \mathrm U(m)}[P_U(\vec n)]^2}{\E_{U \in \mathrm U(m)}[P_U(\vec n)^2]} \approx \frac{M_1(k,n)^2}{M_2(k,n)} = m_2(k,n).
\end{align} 
In the main text, we define various degrees of anticoncentration in terms of the inverse average collision probability $p_2$, which we recall here.
\begin{itemize}
     \item[(A)] 
We say that $P_U, U \in \mathrm U(m)$ \emph{anticoncentrates} if $p_2 = \Omega(1)$. 

\item[(WA)] We say that $P_U$ \emph{anticoncentrates weakly} if $p_2 = \Omega(1/n^a)$ for some $a = O(1)$. 

\item[(NA)] And we say that it \emph{does not anticoncentrate} if $p_2 = O(1/n^{a})$ for any constant $a > 0$. 
 \end{itemize} 

\noindent Here, we motivate those definitions in more detail. 
Clearly (A) implies anticoncentration in the sense of \cref{eq:anticoncentration}, hence the definition.

\paragraph{Lack of anticoncentration (NA)}
Ignoring the average over unitaries, $p_2$ upper-bounds the support of the distribution by $p_2 |\Omega|$, as the maximum-entropy state is the uniform distribution. 
Let us assume for simplicity that $p_{2}$ is actually exponentially small. An exponentially small value of $p_2$ implies that the average support of the outcome distributions $P_U$ is exponentially small, implying that at least a constant fraction (over $U$) of the distributions $P_U$ has exponentially small support, and conversely exponentially larger than uniform probabilities on that support.    
At least for those distributions, this implies an exponentially larger error tolerance compared to $1/|\Omega|$. 
Such an exponentially larger error tolerance makes the approximate average-case hardness conjecture significantly stronger, presumably even untenable. 

While it is possible that for a constant fraction of the $U$ we are in this scenario (see Sec.~V.C of Ref.~\cite{deshpande_tight_2022} for an example), while for another constant fraction, the probabilities are highly spread out, making the anticoncentration property \eqref{eq:anticoncentration} true, this seems like an extremely unlikely state of affairs. 
Indeed, the hiding property implies that it should not matter whether we talk about the distribution over unitaries or over outcomes, which means that the situation described above is a generic feature, rendering \eqref{eq:anticoncentration} false in case $p_2$ is exponentially small. 

\paragraph{Weak anticoncentration (WA)}
Our results show that weak anticoncentration holds in the regime of $k \rightarrow \infty$. But why do we think of a polynomially decaying $p_2$ as \emph{weak} anticoncentration rather than no anticoncentration? 

We argue that this is a meaningful regime in the sense that there is a stronger---but not inconceivable---approximate average-case hardness conjecture associated with the weak anticoncentration regime. 
To see this, observe that weak anticoncentration implies anticoncentration in the sense of \cref{eq:anticoncentration} with $\gamma(\alpha) = \Omega(1/\textrm{poly}(n))$, which means that an inverse polynomial fraction of the outcome probabilities are larger than uniform. 
Technically, using Stockmeyer's algorithm we can thus achieve a multiplicative error for an inverse polynomial fraction of the outcome probabilities. 
To rule out an efficient classical sampler, we thus need to conjecture approximate average-case hardness with constant relative errors for any inverse polynomial fraction of the instances. 
Equivalently, we can formulate a similar conjecture for a polynomially large relative or subexponentially large additive error on a constant fraction.
While clearly much stronger than the requirement of anticoncentration, this is qualitatively different from the lack of anticoncentration scenario (NA), where the difference is superpolynomial. 

\subsection{Anticoncentration of the Exact Distribution}

We also need to show that our definition of anticoncentration allows us to translate between anticoncentration of the approximate distribution based on the hafnians of random Gaussian matrices, which we will refer to as $P_{X}(\vec{n})$, and anticoncentration of the true distribution, $P_{U}(\vec{n})$. For a given output $\vec{n}$, let $\mathcal{D}_{U}$ be the distribution of the symmetric product $U_{1_{k},\vec{n}}^{\top}U_{1_{k},\vec{n}}$ with $U \in \mathrm{U}(m)$. Let $\mathcal{D}_{X}$ be the distribution of the symmetric product $X^{\top}X$ with $X \sim \mathcal{N}(0,1/m)_{c}^{k\times 2n}$. In Conjecture~\ref{con:gaussian}, we conjecture that $\mathcal{D}_{U}$ and $\mathcal{D}_{X}$ become close in total variation distance when $n = o(\sqrt{m})$. However, precisely how close these two distributions are is crucial to whether or not anticoncentration translates between the two output probabily distributions. In what follows, we will refer to anticoncentration in the sense of \cref{eq:anticoncentration} as ``standard'' anticoncentration, and our definition of anticoncentration as ``moment-based.''

Ideally, we would be able to prove that statements about moment-based anticoncentration of $P_{X}(\vec{n})$ imply equivalent statements about moment-based anticoncentration of $P_{U}(\vec{n})$. However, under worst-case assumptions, we can only show that moment-based anticoncentration of $P_{X}(\vec{n})$ implies standard anticoncentration of $P_{U}(\vec{n})$. To understand this, let us fix some notation. Let also $\mathbbm{1}[\cdot]$ be an indicator function which is $1$ if the argument is true and $0$ if it is false. Let $\mathrm{d}\mu$ be the Lebesgue measure on $\mathbb{C}^{2kn}$ (as we consider $k\times 2n$ complex matrices) and $p_{U}(A)$, $p_{X}(A)$ be the respective probabilities of generating $A$ from $\mathcal{D}_{U}$ and $\mathcal{D}_{X}$.

Now, let the total-variation distance between $\mathcal{D}_{U}$ and $\mathcal{D}_{X}$ be less than $\delta$. Then 
\begin{align}
    \Pr_{U \in \mathrm U(m)}[P_{U}(\vec{n}) \geq \epsilon]
    &= \int \mathrm{d}\mu \, p_{U}(A) \mathbbm{1}[P_{A}(\vec{n})\geq \epsilon] \\
    &= \int \mathrm{d}\mu \, (p_{U}(A)-p_{X}(A)+p_{X}(A)) \mathbbm{1}[P_{A}(\vec{n})\geq \epsilon] \\
    &= \int \mathrm{d}\mu \, (p_{U}(A)-p_{X}(A))\mathbbm{1}[P_{A}(\vec{n})\geq \epsilon]+\int \mathrm{d}\mu\,p_{X}(A) \mathbbm{1}[P_{A}(\vec{n})\geq \epsilon] \\
    &\geq -2\delta + \Pr_{X \in \mathcal{G}}[P_{X}(\vec{n}) \geq \epsilon].
\end{align}
In this calculation, we have used the Radon-Nikodym theorem \cite{FollandRealAnalysis2007} to express the probability measures that define $\mathcal{D}_U$ and $\mathcal{D}_{X}$ as $p_{U}(A)\mathrm{d}\mu$ and $p_{X}(A)\mathrm{d}\mu$, respectively.
Therefore
\begin{equation}
    \Pr_{U \in \mathrm U(m)}\left[P_U(\vec n) \geq  \frac{\alpha}{|\Omega_{2n}|}\right] \geq \Pr_{X \in \mathcal{G}}\left[P_{X}(\vec{n}) \geq \frac{\alpha}{|\Omega_{2n}|}\right] - 2\delta  \geq (1- \alpha)^2 \frac 1 {m_{2}(k,n)} - 2\delta. 
\end{equation}
The final step follows from the Paley-Zygmund inequality for the approximate distribution. This proves that we can translate statements on anticoncentration as long as $2\delta$ is smaller than $m_{2}(k,n)^{-1}$, which, as we show in the main text, means $\delta = o(n^{-1/2})$. 

With this in mind, we can make the following more precise version of Conjecture~\ref{con:gaussian} such that, if it holds, moment-based weak anticoncentration of the approximate distribution implies standard weak anticoncentration of the exact distribution:
\addtocounter{conjecture}{-1}
\begin{conjecture}[Formal]\label{con:formal}
Let $\mathcal{D}_{U}$ be the distribution of the symmetric product $U_{1_{k},\vec{n}}^{\top}U_{1_{k},\vec{n}}$ with $U$ unitary and $\vec{n}$ some non-collisional outcome of a Gaussian Boson Sampling experiment. Let $\mathcal{D}_{X}$ be the distribution of the symmetric product $X^{\top}X$ with $X \sim \mathcal{N}(0,1/m)_{c}^{k\times 2n}$. Then, for any $k$ such that $1 \leq k \leq m$, and for any $\delta > 0$ such that $m\geq n^{2}/\delta$,
\begin{equation}
    \mathsf{tvd}(\mathcal{D}_{U}, \mathcal{D}_{X}) = O(\delta).
\end{equation}
Specifically, if $\delta = o(n^{-1/2})$, then $m \geq n^{5/2}$. 
\end{conjecture}
The motivation behind the choice of $m \geq n^{2}/\delta$ is based on the equivalent conjecture for Fock Boson Sampling in Ref.~\cite{aaronsonComputationalComplexityLinear2013}. There, the authors are able to prove the equivalent result for $m \geq n^{5+\epsilon}/\delta$ (for arbitrarily small, constant $\epsilon$), but they suspect that the result can be pushed further to $m \geq n^{2}/\delta$. We note that this choice makes our formal conjecture slightly stronger than the equivalent formal conjecture in Ref.~\cite{deshpandeQuantumComputationalAdvantage2022a}. 

As we have shown, in order to translate our results on moment-based weak anticoncentration from the approximate to the true distribution in the worst case, we require $\delta = o(n^{-1/2})$. Therefore, in order to translate statements about anticoncentration, the formal version of our conjecture requires $m \geq n^{5/2}$. 

However, it is worth noting that we do not believe that this worst-case scenario truly reflects the way in which $P_{X}(\vec{n})$ approaches $P_{U}(\vec{n})$, i.e., where all of the error is concentrated on a single probability. In general, the intuition is that if hiding holds, then it is more likely that the errors are more evenly distributed amongst all of the exponentially many output probabilities. Using this intuition, each individual probability only receives an error of approximately $\delta/|\Omega_{2n}|$. If this is true, then we can show that moment-based weak anticoncentration of $P_{X}(\vec{n})$ does actually imply the same for $P_{U}(\vec{n})$. Specifically, say that $P_{U}(\vec{n}) \approx P_{X}(\vec{n}) \pm \delta/|\Omega_{2n}| \approx P_{X}(\vec{n}) \pm \delta \E[P_{X}(\vec{n})]$ (as per \cref{sec:asymptotics}). Then
\begin{align}
 \frac{\E[P_{U}(\vec{n})^{2}]}{(\E[P_{U}(\vec{n})])^{2}} &\approx \frac{\E[(P_{X}(\vec{n}) \pm \delta \E[P_{X}(\vec{n})])^{2}]}{(\E[P_{X}(\vec{n}) \pm \delta \E[P_{X}(\vec{n})]])^{2}}\\
 &= \frac{\E[P_{X}(\vec{n})^{2}]\pm2\delta\E[P_{X}(\vec{n})]^{2}+\delta^{2}\E[P_{X}(\vec{n})]^{2}}{(1\pm\delta)^{2}\E[P_{X}(\vec{n})]^{2}} \\
 &\approx \frac{1}{(1\pm\delta)^{2}}  \frac{\E[P_{X}(\vec{n})^{2}]}{(\E[P_{X}(\vec{n})])^{2}} + \frac{\pm 2\delta + \delta^{2}}{(1\pm \delta)^{2}} \\
 &= \frac{1}{(1\pm\delta)^{2}}\frac{\E[P_{X}(\vec{n})^{2}]}{(\E[P_{X}(\vec{n})])^{2}} + 1 - \frac{1}{(1\pm \delta)^{2}} \\
 &\leq \frac{1}{(1-\delta)^{2}}  \frac{\E[P_{X}(\vec{n})^{2}]}{(\E[P_{X}(\vec{n})])^{2}}+1.
\end{align}
In our case, where the normalized second moment of $P_{X}(\vec{n})$ scales at least polynomially in $n$, and $\delta$ scales inverse polynomially in $n$, weak anticoncentration or lack of anticoncentration of $P_{X}(\vec{n})$ in terms of the normalized second moment adequately translates to $P_{U}(\vec{n})$ as well. Note that, in this case, we are assuming that $\delta$, which is the total variation distance between the distributions of matrices, extends to a bound on the total variation distance between the probabilities themselves. This intuitively arises from the fact that any map from the distribution of the matrices to probabilities must be bounded, meaning we can translate the total variation distance from one to the other (however, formalizing this would require dealing with some subtleties induced by the fact that the hafnian of a product of Gaussians is not technically bounded, but any large hafnians only arise with extremely small probabilities).

\section{Scattershot Boson Sampling Explanation of the Transition in Anticoncentration}

In Scattershot Boson Sampling (SBS), the setup is as follows. 
$m = \omega(n^2)$ two-mode squeezed states with squeezing parameter $r$ are prepared. 
The photon number distribution of the two-mode squeezed states is supported on Fock states of the form $\ket n \ket n$ for $n \in \mathbb N_0$. 
One half of each two-mode squeezed state is then measured in the Fock basis, yielding, with high probability, an outcome $n_i \in \{0,1\}$ (assuming $r$ is small enough). 
Collecting outcomes in the vector $\vec n = (n_1, \ldots, n_m)$, the other half of the input modes is now in the postselected state $\ket {\vec n} = \bigotimes_{i=1}^m \ket {n_i}$. 
The outcome probabilities after passing this input state through the linear optical unitary $U$ and measuring in the Fock basis yielding outcome $\vec o = (o_1, \ldots, o_m)$, $o_i \in \{0,1\}$ is then given by 
\begin{align}
    P_U(\vec n, \vec o) = |\per(U_{\vec n, \vec o})|^2
\end{align}
of the submatrix $U_{\vec n, \vec o}$ in which we select the rows and columns according to the indices with nonzero entries in $\vec n$ and $\vec o$.
But conditioned on input and output states being collision-free and the hiding property, the distribution of matrices $U_{\vec n, \vec o}$ equals that of the Boson Sampling submatrices $U_{\vec 1_n, \vec o}$, where the photons in the input state are by convention in the first $n$ modes. 
The properties of Scattershot Boson Sampling postselected on collision-free outcomes in a fixed photon number sector are therefore equal to the properties of standard Boson Sampling. 

We now argue that this equivalency hinges essentially on the fact that at least $\omega(n^2)$ of the input modes are squeezed. 
To this end, consider a modification of Scattershot Boson Sampling in which only $k$ out of the $m$ modes are prepared in one half of a two-mode squeezed state, while the remaining $m-k$ modes are prepared in the vacuum state. 
This closely resembles the GBS setting, of course. 
Let us also consider a squeezing parameter $r$ of every two-mode squeezed state chosen such that the mean photon number after postselection is given by $n$. 
To achieve this, we pick the mean photon number per mode, which is given by $\sinh^2(r)$ to be equal $n/k$ to obtain a total of $k \sinh^2 r = n$ photons on average. 
This ensures that in the postselection we end up with $n$ photons with high probability. 

Recall that a two-mode squeezed vacuum state with squeezing parameter $r$ and phase $\phi$ has a Fock expansion given by
\begin{equation}
    \ket{\mathrm{TMSV}} = \frac{1}{\cosh(r)}\sum_{\ell = 0}^{\infty} (-e^{i\phi}\tanh r)^{n}\ket{nn},
\end{equation}
thus leading to a probability of measuring $\ell$ photons in one mode of $\tanh^{2\ell}r/\cosh^{2}r$.
Therefore, if the input consists of $k$ two-mode squeezed vacuum states, then the probability that, after measuring one half of each state, one observes a collision is
\begin{equation}
    \Pr[\text{collision}] = 1-\left(\frac{1}{\cosh^{2}r} + \frac{\tanh^{2}r}{\cosh^{2}r}\right)^{k} = 1-\left(\frac{1}{1+n/k} + \frac{n/k}{(1+n/k)^{2}}\right)^{k} = 1 - \left(1-\left(\frac{n/k}{1+n/k}\right)^{2}\right)^{k}. 
\end{equation}
We can rewrite this via Taylor series as
\begin{equation}
    \Pr[\text{collision}] = 1-\exp\left(-\frac{n^{2}}{k} + kO(n/k)^{3})\right)
\end{equation}
assuming $k = \omega(n)$. This collision probability remains lower bounded by a constant for $k = O(n^2)$, but vanishes for any $k = \omega(n^2)$. 
Thus, the probability of a collision in the input state of SBS remains high until $k = \Theta(n^2)$ and decays then. 
But because in SBS the roles of the (postselected) input state and the output state are symmetric, a collision implies a failure of hiding and, therefore, a failure of anticoncentration in the regime $k = O(n^2)$.  
Conversely, for $k = \omega(n^2)$ we believe that hiding holds \cite{jiang_how_2006,aaronsonComputationalComplexityLinear2013}, and hence Lemma 8.8 of \textcite{aaronsonComputationalComplexityLinear2013} shows weak anticoncentration for SBS with the inverse average collision probability $p_2 = 1/n $ in this regime. 

This shows that generalized SBS with a variable number of input squeezed states undergoes a transition in anticoncentration as we find it here for the case of GBS. 
It is not at all clear that the transition in SBS implies a transition in GBS, however, as GBS does not involve postselection. 
Indeed, in SBS, the anticoncentration coincides with---or rather \emph{is}---a transition in the hiding property. 
In GBS, in contrast, hiding is conjectured to hold for all $k$, while we do see the transition in anticoncentration.
The situation in GBS is not immediately comparable to that in this modified SBS scenario because the input single-mode squeezed states are supported on even numbers of photons, and therefore any nonzero photon number input states are collision-full. Therefore, as mentioned in the discussion in the main text, the possible connections outlined here deserve future consideration. 
\end{document}